\documentclass[11pt,a4paper]{article}
\usepackage{amssymb}
\usepackage{amsmath}
\usepackage{amsfonts}
\usepackage{dsfont}
\usepackage{amsthm}
\usepackage{mathrsfs}
\usepackage{hyperref}
\usepackage{color}
\usepackage[margin=2.41cm]{geometry}
\usepackage[all,cmtip]{xy}
\usepackage[utf8]{inputenc}
\usepackage{graphicx}
\usepackage{varwidth}
\usepackage{comment}

\usepackage{upgreek}
\usepackage{rotating}

\usepackage{tikz}
\usetikzlibrary{shapes.geometric}

\usepackage[shortlabels]{enumitem}

\definecolor{darkred}{rgb}{0.8,0.1,0.1}
\hypersetup{
     colorlinks=false,         % false: boxed links; true: colored links,false is default
     linkcolor=darkred,
     citecolor=blue,
}

\theoremstyle{plain}
\newtheorem{theo}{Theorem}[section]
\newtheorem{lem}[theo]{Lemma}
\newtheorem{propo}[theo]{Proposition}
\newtheorem{cor}[theo]{Corollary}

\theoremstyle{definition}
\newtheorem{defi}[theo]{Definition}

\newtheorem{problem}[theo]{Open Problem}

\newtheorem{exx}[theo]{Example}
\newenvironment{ex}
  {\pushQED{\qed}\exx}
  {\popQED\endexx}

\newtheorem{remm}[theo]{Remark}
\newenvironment{rem}
  {\pushQED{\qed}\remm}
  {\popQED\endremm}

\numberwithin{equation}{section}

\def\nn{\nonumber}

\def\bbK{\mathbb{K}}
\def\bbR{\mathbb{R}}
\def\bbC{\mathbb{C}}

\def\bbZ{\mathbb{Z}}

\def\ii{{\,{\rm i}\,}}

\def\Aut{\mathrm{Aut}}

\def\id{\mathrm{id}}
\def\Id{\mathrm{Id}}
\def\supp{\mathrm{supp}}

\def\dd{\mathrm{d}}

\def\cc{\mathrm{c}}

\def\vc{\mathrm{vc}}
\def\vpc{\mathrm{vpc}}
\def\vfc{\mathrm{vfc}}

\def\1{\mathbf{1}}
\def\oone{\mathds{1}}
\def\op{\mathrm{op}}

\def\pr{\mathrm{pr}}

\def\ver{\mathrm{v}}
\def\data{\mathrm{data}}
\def\solve{\mathrm{solve}}

\def\Loc{\mathbf{Loc}}

\def\Man{\mathbf{Man}}

\def\AQFT{\mathbf{AQFT}}
\def\2AQFT{\mathbf{2AQFT}}
\def\Fun{\mathbf{Fun}}
\def\Open{\mathbf{Open}}

\def\Sh{\mathbf{Sh}}
\def\St{\mathbf{St}}

\def\VecBun{\mathbf{VecBun}}
\def\astMon{{}^\ast\mathbf{Mon}}

\def\Set{\mathbf{Set}}
\def\Alg{\mathbf{Alg}}
\def\astAlg{{}^\ast\mathbf{Alg}_\bbC}

\def\Vec{\mathbf{Vec}}
\def\Ch{\mathbf{Ch}}

\def\astObj{{}^\ast\mathbf{Obj}}
\def\Pois{\mathbf{PoVec}_\bbR}
\def\IPVec{\mathbf{IPVec}_\bbC}

\def\Map{\mathsf{Map}}

\def\Cat{\mathbf{Cat}}

\def\SMCat{\mathbf{SMCat}}

\def\AAA{\mathfrak{A}}
\def\BBB{\mathfrak{B}}
\def\LLL{\mathfrak{L}}
\def\BBB{\mathfrak{B}}

\def\WWW{\mathfrak{W}}
\def\Sol{\mathsf{Sol}}
\def\CCR{\mathfrak{CCR}}
\def\CAR{\mathfrak{CAR}}

\def\colim{\mathrm{colim}}

\def\bicolim{\mathrm{bicolim}}

\newcommand\und[1]{\underline{#1}}
\newcommand\ovr[1]{\overline{#1}}
\newcommand\mycom[2]{\genfrac{}{}{0pt}{}{#1}{#2}}

\def\sk{\vspace{1mm}}

\makeatletter
\let\@fnsymbol\@alph
\makeatother

\title{%
Smooth $1$-dimensional algebraic quantum field theories
}

\author{%
Marco Benini$^{1,2,a}$, 
Marco Perin$^{3,b}$\ and\ 
Alexander Schenkel$^{3,c}$\vspace{4mm}\\
{\small ${}^1$ Dipartimento di Matematica, Universit\`a di Genova,}\\
{\small Via Dodecaneso 35, 16146 Genova, Italy.}\vspace{2mm}\\
{\small ${}^2$ INFN, Sezione di Genova,}\\
{\small Via Dodecaneso 33, 16146 Genova, Italy.}\vspace{2mm}\\
{\small ${}^3$ School of Mathematical Sciences, University of Nottingham,}\\
{\small University Park, Nottingham NG7 2RD, United Kingdom.}\vspace{4mm}\\
{\small \begin{tabular}{ll}
Email: & ${}^a$~\texttt{benini@dima.unige.it}\\
& ${}^b$~\texttt{marco.perin@nottingham.ac.uk}\\
& ${}^c$~\texttt{alexander.schenkel@nottingham.ac.uk}\vspace{2mm}
\end{tabular}
}
}

\date{October 2021}

\begin{document}

\maketitle

\vspace{-5mm}

\begin{abstract}
\noindent This paper proposes a refinement of the usual concept of algebraic quantum field theories (AQFTs) to theories that are smooth in the sense that they assign to every smooth family of spacetimes a smooth family of observable algebras. Using stacks of categories, this proposal is realized concretely for the simplest case of $1$-dimensional spacetimes, leading to a stack of smooth $1$-dimensional AQFTs. Concrete examples of smooth AQFTs, of smooth families of smooth AQFTs and of equivariant smooth AQFTs are constructed. The main open problems that arise in upgrading this approach to higher dimensions and gauge theories are identified and discussed.
\end{abstract}

\vspace{-1mm}

\paragraph*{Keywords:} algebraic quantum field theory, stacks of categories, vertical geometry of fiber bundles, smoothly parametrized differential equations
\vspace{-2mm}

\paragraph*{MSC 2020:} 81Txx, 18F20, 18N10
\vspace{-1mm}

\tableofcontents
%\renewcommand{\baselinestretch}{1.0}\normalsize

%\newpage

%%%%%%%%%%%%%%%%%%%%%%%%%%%%%%%%%%%%%%%%%%%%%%%%
%%%%%%%%%%%%%%%%%%%%%%%%%%%%%%%%%%%%%%%%%%%%%%%%

\section{\label{sec:intro}Introduction and summary}
An $m$-dimensional algebraic quantum field theory (AQFT) is a functor
$\AAA : \Loc_m\to \astAlg$ from a suitable category of $m$-dimensional
Lorentzian spacetimes to the category of associative and unital 
$\ast$-algebras over $\bbC$. The algebra $\AAA(M)$ that is assigned by
this functor to a spacetime $M$ is interpreted as the algebra of 
quantum observables of the theory $\AAA$ that can be measured in $M$.
Such functors are also required to satisfy a list of physically motivated axioms,
see e.g.\ \cite{BFV,FewsterVerch}, which includes most notably the 
Einstein causality axiom.
\sk

Even though this axiomatic definition of AQFTs is by now widely used in
the relevant research community and has led to interesting 
model-independent results, we would like to point out 
the following issue that is usually not discussed: 
Suppose that we consider a 
family of spacetimes $\{M_s\in\Loc_m\}_{s\in\bbR}$ that
depends smoothly (in some appropriate sense as explained in this paper) 
on a parameter $s\in\bbR$. This $s$-dependence could, for example,
be due to changing smoothly the coefficients of the metric tensor.
Evaluating an AQFT $\AAA : \Loc_m\to \astAlg$ on this smooth family
results in a family of algebras $\{\AAA(M_s)\in \astAlg\}_{s\in\bbR}$
which, however, will in general not be smooth in any appropriate sense
because smoothness is not covered by the usual AQFT axioms.
Similarly, given a smooth family $\{f_s : M_s\to N_s\}_{s\in\bbR}$
of spacetime morphisms, the associated family 
$\{\AAA(f_s) : \AAA(M_s) \to \AAA(N_s)\}_{s\in\bbR}$ of 
$\astAlg$-morphisms will in general not be smooth.
In our opinion, encoding a suitable concept of 
smoothness for these families as part of the axioms of 
AQFT is desirable for several reasons:
1.)~Physically speaking, smoothness of these families
means that a small variation at the level spacetimes 
does not have a too drastic effect on the 
observable content of the theory, hence it excludes
models with unpleasant discontinuous behavior.
2.)~Certain standard constructions in AQFT, such as the
computation of the stress-energy tensor as a derivative
of the relative Cauchy evolution \cite{BFV}, only
exist for models that react sufficiently 
smoothly to metric perturbations. 
3.)~In the context of AQFT on spacetimes with
background fields, see e.g.\ \cite{Zahn}, such a 
smooth dependence may be used to describe 
an adiabatic switching of the interaction with the background fields.
Similarly, it enables us to introduce AQFTs that are smoothly equivariant
with respect to an action of a Lie group.
\sk

The main goal of this paper is to make some first steps towards developing
a refinement of the axiomatic foundations of AQFT that encodes the 
preservation of smooth families as part of its structure.
The key idea behind our approach is to refine the ordinary
categories $\Loc_m$ and $\astAlg$ that enter the definition
of AQFTs to {\em stacks of categories}, which will encode
precise concepts of smooth families of spacetimes and algebras,
and to introduce a concept of {\em smooth AQFTs} in terms 
of morphisms between these stacks. A similar 
program of smooth refinements of QFTs has been 
developed successfully within other approaches, 
in particular for functorial QFTs in the sense
of Atiyah and Segal \cite{StolzTeichner,Pavlov,BunkWaldorf,LudewigStoffel},
however as of now this idea seems to be unexplored in the context of AQFT.
In order to outline our proposal in the simplest possible setting
and to circumvent in this first paper certain technical challenges (of both analytical and algebraic nature, 
see Section \ref{sec:higher}), we consider only the simplest
case of dimension $m=1$, which physically
represents AQFTs on time intervals (i.e.\ quantum mechanics).
We believe that, due to its simplicity, the case of $1$-dimensional
AQFTs is perfectly suited to explain the main ideas and features of 
our proposed framework for smooth AQFTs and to illustrate this formalism
through the simplest possible examples.
\sk

Our framework for smooth AQFTs introduces naturally
a second layer of smoothness. Because we realize smooth $1$-dimensional
AQFTs as the points of a stack $\AQFT_1^\infty$, we can also make precise 
sense of questions like what are ``smooth families of smooth AQFTs''
and in particular what are ``smooth curves of smooth AQFTs''.
We shall illustrate through simple examples that smooth variations (e.g.\ an adiabatic switching)
of the external parameters of a theory, such as the mass parameter,
define such smooth families of smooth AQFTs.
Furthermore, we show that each smooth AQFT has an associated 
smooth automorphism group, refining the discrete automorphism 
groups of ordinary AQFTs \cite{FewsterAut}, and explain how these
are related to smooth AQFTs that are equivariant with respect to 
a smooth action of a Lie group. 
We construct a concrete example of the latter that captures the global 
$U(1)$-symmetry of the $1$-dimensional massless Dirac field. 
\sk

The outline of the remainder of this paper is as follows: 
Section \ref{sec:stacks} contains a brief review of some relevant
aspects of the theory of stacks of categories that we shall need in 
this work. In Section \ref{sec:definition} we introduce the 
stacks of categories $\Loc_1^\infty$ and $\astAlg^\infty$ that provide 
smooth refinements of the category $\Loc_1$ of $1$-dimensional spacetimes
and of the category $\astAlg$ of associative and unital $\ast$-algebras.
The stack of smooth $1$-dimensional AQFTs is then defined
as the mapping stack $\AQFT_1^\infty:=\Map(\Loc_1^\infty,\astAlg^\infty)$
and we shall explore some interesting consequences of this definition, 
including a natural notion of smooth automorphism group of a smooth AQFT 
and its relation to $G$-equivariant smooth AQFTs, for $G$ a Lie group.
Section \ref{sec:CCR-CAR} develops two stack morphisms $\CCR$ and $\CAR$ 
that are smooth refinements of the usual canonical (anti-)commutation
relation quantization functors for Bosonic (resp., Fermionic) theories. 
These are later used for constructing explicit examples 
in Section \ref{sec:examples}, which illustrate 
our proposed approach to smooth AQFT. 
In Subsection \ref{subsec:Green}, we introduce smooth refinements of retarded/advanced 
Green operators and prove their existence in simple cases 
through explicit formulas. 
We then construct in Subsection \ref{subsec:example} 
a concrete example of a smooth family of smooth $1$-dimensional AQFTs, 
which can be interpreted physically
as (a smooth analog of) the $1$-dimensional massive scalar field 
(quantum harmonic oscillator) 
in the presence of a smooth variation of the mass (frequency) parameter.
In Subsection \ref{subsec:fermion}, we construct
a concrete example of a $U(1)$-equivariant smooth $1$-dimensional
AQFT, which can be interpreted physically
as (a smooth analog of) the $1$-dimensional massless Dirac field, 
together with its global $U(1)$-symmetry.
Section \ref{sec:higher} provides a concise list of open problems
that have to be solved to upgrade our approach 
to encompass higher dimensions $m\geq 2$ 
and gauge theories. Most pressingly,
Open Problem \ref{problem1} poses the question of existence
of smoothly parametrized retarded/advanced Green operators for vertical normally hyperbolic operators
on smooth families of Lorentzian spacetimes, which goes beyond the standard results
developed e.g.\ in \cite{BGP} and might be of interest to researchers in hyperbolic PDE theory.

%%%%%%%%%%%%%%%%%%%%%%%%%%%%%%%%%%%%%%%%%%%%%%%%

\section{\label{sec:stacks}Preliminaries on stacks of categories}
We shall briefly review some basic concepts  
from the theory of {\em stacks of categories} that we need
to describe a smooth refinement of algebraic 
quantum field theories (AQFTs). Our perspective on smoothness is through the functor
of points approach, see e.g.\ \cite[Section 3.2]{BSreview} and \cite{BSCahiers}
for introductions in the context of AQFT
and also \cite{Schreiber} for a more detailed overview. We also refer to 
\cite{Leinster,Lack} for the relevant $2$-categorical background
and to \cite{Vistolli} for a detailed introduction to the theory of stacks.
\sk

Let $\Man$ denote the category of (finite-dimensional) smooth manifolds and smooth maps.
The usual open cover Grothendieck topology endows $\Man$ with the structure of a site. 
We choose the site $\Man$ because our aim is 
to formalize {\it smooth} families. The framework of stacks is however 
very flexible and can be adapted to model other types of families, 
such as continuous or algebraic, by an appropriate choice of site. 
\begin{defi}\label{def:prestack}
A {\em prestack} (of categories) is a pseudo-functor
$X : \Man^\op\to \Cat$ to the $2$-category $\Cat$ 
of categories, functors and natural transformations.
Explicitly, this consists of the following data:
\begin{enumerate}[(1)]
\item For each object $U\in \Man$, a category $X(U)$.
\item For each morphism $h : U\to U^\prime$ in $\Man$, a functor
$X(h) : X(U^\prime)\to X(U)$.
\item For each pair of composable morphisms $h : U\to U^\prime$
and $h^\prime : U^\prime\to U^{\prime\prime}$ in $\Man$, a natural
isomorphism $X_{h^\prime,h} : X(h)\,X(h^\prime)\Rightarrow X(h^\prime\,h)$
of functors from $X(U^{\prime\prime})$ to $X(U)$.
\item For each object $U\in \Man$, a natural isomorphism
$X_{U} : \id_{X(U)} \Rightarrow X(\id_U)$ of functors from $X(U)$ to $X(U)$.
\end{enumerate}
These data have to satisfy the following axioms:
\begin{enumerate}[(i)]
\item For all triples of composable morphisms 
$h : U\to U^\prime$, $h^\prime : U^\prime\to U^{\prime\prime}$
and $h^{\prime\prime} : U^{\prime\prime}\to U^{\prime\prime\prime}$ in $\Man$,
the diagram
\begin{flalign}
\xymatrix@C=5em{
\ar@{=>}[d]_-{\Id \ast X_{h^{\prime\prime},h^\prime}}X(h) \,X(h^{\prime})\, X(h^{\prime\prime}) \ar@{=>}[r]^-{X_{h^\prime,h}\ast \Id} 
~&~ X(h^\prime\,h)\, X(h^{\prime\prime})\ar@{=>}[d]^-{X_{h^{\prime\prime},h^\prime h}}\\
X(h) \,X(h^{\prime\prime}\,h^{\prime}) \ar@{=>}[r]_-{X_{h^{\prime\prime}h^\prime,h}} ~&~ X(h^{\prime\prime} \,h^\prime\, h)
}
\end{flalign}
of natural transformations commutes. (The capital $\Id$ denotes identity natural transformations
and $\ast$ denotes horizontal composition of natural transformations.)

\item For all morphisms $h: U\to U^\prime$ in $\Man$, the two diagrams
\begin{flalign}
\xymatrix{
\id_{X(U)}\,X(h)\ar@{=>}[d]_-{X_{U}\ast \Id} \ar@{=}[dr]~&~ ~&~ X(h)\,\id_{X(U^\prime)} \ar@{=>}[d]_-{\Id \ast X_{U^\prime}} \ar@{=}[dr]~&~\\
X(\id_U)\,X(h)\ar@{=>}[r]_-{X_{h,\id_U}}~&~X(h) ~&~X(h)\,X(\id_{U^\prime}) \ar@{=>}[r]_-{X_{\id_{U^\prime},h}}~&~X(h)
}
\end{flalign}
of natural transformations commute.
\end{enumerate}
\end{defi}
\begin{rem}
From now on we shall often follow the usual convention 
of suppressing the symbols $X_{h^\prime,h}$ and $X_U$ 
denoting the coherence isomorphisms of a prestack $X$. 
Definition \ref{def:prestack} should help readers unfamiliar 
with this convention to extrapolate from the context 
which coherence isomorphism is relevant at any given point. 
\end{rem}

Given any prestack $X: \Man^{\op}\to \Cat$, one can define, for every
manifold  $U\in\Man$ and every open cover $\{U_\alpha \subseteq U\}$,
the associated {\em descent category} $X(\{U_\alpha\subseteq U\})\in\Cat$, 
see e.g.\ \cite[Definition 4.2]{Vistolli}.
An object in this category is a tuple
\begin{subequations}\label{eqn:descentobject}
\begin{flalign}\label{eqn:descentobject1}
\Big(\big\{x_\alpha \big\},\big\{ \varphi_{\alpha\beta} : x_\beta\vert_{U_{\alpha\beta}}
\to x_\alpha\vert_{U_{\alpha\beta}}\big\}\Big)\,\in\, X(\{U_\alpha\subseteq U\})
\end{flalign}
of families of objects $x_\alpha\in X(U_\alpha)$ and isomorphisms 
$\varphi_{\alpha\beta}$ in $ X(U_{\alpha\beta})$ satisfying
\begin{flalign}\label{eqn:descentobject2}
\xymatrix@C=1.2em@R=1.6em{
~&~x_{\beta}\vert_{U_{\beta\gamma}}\vert_{U_{\alpha\beta\gamma}}  \ar[rd]^-{\cong}~&~ ~&~ ~&~ ~&~\\
\ar[d]_-{\cong} x_{\gamma}\vert_{U_{\beta\gamma}}\vert_{U_{\alpha\beta\gamma}} \ar[ru]^-{\varphi_{\beta\gamma}\vert_{U_{\alpha\beta\gamma}}~~}~&~~&~
 x_{\beta}\vert_{U_{\alpha\beta}}\vert_{U_{\alpha\beta\gamma}}\ar[d]^-{\varphi_{\alpha\beta}\vert_{U_{\alpha\beta\gamma}}} 
 ~&~ ~&~ \ar[d]_-{\cong }x_\alpha\vert_{U_{\alpha\alpha}} \ar[r]^-{\varphi_{\alpha\alpha}} ~&~ x_\alpha\vert_{U_{\alpha\alpha}} \ar[d]^-{\cong}\\
\ar[dr]_-{\varphi_{\alpha\gamma}\vert_{U_{\alpha\beta\gamma}} }x_{\gamma}\vert_{U_{\alpha\gamma}}\vert_{U_{\alpha\beta\gamma}} 
~&~ ~&~x_{\alpha}\vert_{U_{\alpha\beta}}\vert_{U_{\alpha\beta\gamma}}\ar[dl]^-{\cong} ~&~ ~&~  x_\alpha \ar[r]_-{\id_{x_\alpha}}~&~x_\alpha\\
~&~x_{\alpha}\vert_{U_{\alpha\gamma}}\vert_{U_{\alpha\beta\gamma}}  ~&~ ~&~ ~&~~&~ 
}
\end{flalign}
\end{subequations}
for all $\alpha,\beta,\gamma$. Here we denoted by 
$U_{\alpha_1\alpha_2\cdots\alpha_n} := U_{\alpha_1}\cap U_{\alpha_2}\cap\cdots \cap U_{\alpha_n}$
the intersection of open subsets and by $\vert_{\tilde{U}} := X(\iota_{\tilde{U}}^U) : X(U)\to X(\tilde{U})$  the functor associated 
with a subset inclusion morphism $\iota_{\tilde{U}}^U : \tilde{U} \to U$ in $\Man$.
The unlabeled isomorphisms in \eqref{eqn:descentobject2} are given by the 
coherence isomorphisms associated with the 
pseudo-functor $X : \Man^\op\to \Cat$. A morphism 
\begin{subequations}\label{eqn:descentmorphism}
\begin{flalign}\label{eqn:descentmorphism1}
\{\psi_\alpha\} \,:\, \big(\{x_\alpha\},\{\varphi_{\alpha\beta}\}\big)~\longrightarrow~
\big(\{x^\prime_\alpha\},\{\varphi^\prime_{\alpha\beta}\}\big)
\end{flalign}
in the descent category $X(\{U_\alpha\subseteq U\})$ 
is a family of morphisms $\psi_\alpha : x_\alpha\to x_\alpha^\prime$ in $X(U_\alpha)$ satisfying
\begin{flalign}\label{eqn:descentmorphism2}
\xymatrix@C=2.4em@R=1.6em{
\ar[r]^-{\psi_\beta\vert_{U_{\alpha\beta}}}x_\beta\vert_{U_{\alpha\beta}} \ar[d]_-{\varphi_{\alpha \beta}}~&~ x^\prime_\beta\vert_{U_{\alpha\beta}}\ar[d]^-{\varphi^\prime_{\alpha\beta}}\\
x_\alpha\vert_{U_{\alpha\beta}} \ar[r]_-{\psi_{\alpha}\vert_{U_{\alpha\beta}}}~&~x^\prime_\alpha\vert_{U_{\alpha\beta}}
}
\end{flalign}
\end{subequations}
for all $\alpha,\beta$. There exists a canonical functor
\begin{flalign}\label{eqn:functordescentcategory}
\nn X(U)~&\longrightarrow~X(\{U_\alpha\subseteq U \})\quad,\\
\nn x~&\longmapsto~ \big(\{x\vert_{U_\alpha}\}, \{x\vert_{U_\beta}\vert_{U_{\alpha\beta}} 
\stackrel{\cong}{\longrightarrow} x\vert_{U_\alpha}\vert_{U_{\alpha\beta}}\}\big)\quad,\\
\psi~& \longmapsto ~\{\psi\vert_{U_{\alpha}}\}\quad.
\end{flalign}
\begin{defi}\label{def:stack}
A {\em stack} (of categories) is a prestack $X:\Man^\op\to \Cat$ 
that satisfies the following descent condition: For every $U\in \Man$ and every 
open cover $\{U_\alpha\subseteq U\}$, the functor \eqref{eqn:functordescentcategory} is 
an equivalence of categories. 
\end{defi}
\begin{ex}\label{ex:vecbun}
A simple example is the stack of vector bundles
$\VecBun_{\bbK} : \Man^\op\to \Cat$, where $\bbK$ denotes either 
the field of real $\bbR$ or complex $\bbC$ numbers.  It assigns to
each manifold $U\in\Man$ the category
$\VecBun_\bbK(U)$ of (locally trivializable and finite rank) $\bbK$-vector bundles
over $U$, with morphisms given by vector bundle maps 
that preserve the base space. To a morphism $h: U\to U^\prime$ in $\Man$
it assigns the functor $h^\ast : \VecBun_\bbK(U^\prime)\to\VecBun_\bbK(U)$
that forms pullback bundles. The coherence isomorphisms are canonically given 
by the universal property of pullback bundles. The descent condition
expresses the following local-to-global (or gluing)
properties of vector bundles and their morphisms:
Vector bundles on $U$ can be described equivalently in terms of families 
of vector bundles on an open cover $\{U_\alpha\subseteq U\}$ 
and transition functions on the overlaps $U_{\alpha\beta}$, 
see e.g.\ \cite[Problem 10-6]{Lee}. (Observe that the diagrams in \eqref{eqn:descentobject2}
are the cocycle conditions for the transition functions.)
From this perspective, vector bundle maps on $U$ can be described equivalently
in terms of families of vector bundle maps on an open cover $\{U_\alpha\subseteq U\}$ 
that are compatible (in the sense of \eqref{eqn:descentmorphism2}) with the transition functions.
\end{ex}

\begin{defi}\label{def:stackmorphism}
A {\em morphism} $F : X\to Y$ between two stacks is a pseudo-natural transformation
between the underlying pseudo-functors $X,Y :\Man^\op\to \Cat$. Explicitly, this consists of the following data:
\begin{enumerate}[(1)]
\item For each $U\in\Man$, a functor $F_U : X(U)\to Y(U)$.

\item For each morphism $h: U\to U^\prime$ in $\Man$, a natural isomorphism
\begin{flalign}
\xymatrix@C=2em@R=2em{
\ar[d]_-{X(h)}X(U^\prime) \ar[r]^-{F_{U^\prime}}~&~ Y(U^\prime)\ar[d]^-{Y(h)} \ar@{=>}[dl]_-{F_h}\\
X(U) \ar[r]_-{F_{U}}~&~Y(U)
}
\end{flalign}
\end{enumerate}
These data have to satisfy the following axioms:
\begin{enumerate}[(i)]
\item For all pairs of composable morphisms $h : U\to U^\prime$ and $h^\prime : U^\prime\to U^{\prime\prime}$
in $\Man$, the diagram
\begin{flalign}
\xymatrix@C=5em{
 \ar@{=>}[d]_-{Y_{h^\prime,h}\ast\Id} Y(h)\,Y(h^\prime)\,F_{U^{\prime\prime}} \ar@{=>}[r]^-{\Id \ast F_{h^\prime}}~&~ 
 Y(h)\,F_{U^\prime}\,X(h^\prime)\ar@{=>}[r]^-{F_{h}\ast \Id } ~&~
 F_{U}\,X(h)\,X(h^\prime) \ar@{=>}[d]^-{\Id \ast X_{h^{\prime},h}}\\
Y(h^\prime\,h) F_{U^{\prime\prime}} \ar@{=>}[rr]_-{F_{h^\prime h}}~&~  ~&~  F_{U}\,X(h^\prime\, h)
}
\end{flalign}
of natural transformations commutes.

\item For all $U\in \Man$, the diagram
\begin{flalign}
\xymatrix@C=3em{
\ar@{=>}[d]_-{Y_U\ast \Id} \id_{Y(U)} \, F_U \ar@{=}[r] ~&~ F_U \,\id_{X(U)}\ar@{=>}[d]^-{\Id\ast X_U}\\
Y(\id_U)\,F_U \ar@{=>}[r]_-{F_{\id_U}}~&~ F_U\,X(\id_U)
}
\end{flalign}
of natural transformations commutes.
\end{enumerate}
\end{defi}

\begin{defi}\label{def:stack2morphism}
A {\em $2$-morphism} $\zeta : F\Rightarrow G$ between two stack morphisms 
$F,G : X\to Y$ is a modification between the underlying pseudo-natural 
transformations. Explicitly, this consists of the following data:
\begin{enumerate}[(1)]
\item For each $U\in\Man$, a natural transformation
$\zeta_U : F_U \Rightarrow G_U$ of functors from $X(U)$ to $Y(U)$.
\end{enumerate}
These data have to satisfy the following axioms:
\begin{itemize}
\item[(i)] For all morphisms $h:U\to U^\prime$ in $\Man$, the diagram
\begin{flalign}
\xymatrix@C=5em{
\ar@{=>}[d]_-{F_h} Y(h)\,F_{U^\prime} \ar@{=>}[r]^-{\Id \ast \zeta_{U^\prime}} ~&~ Y(h)\,G_{U^\prime}\ar@{=>}[d]^-{G_h}\\
F_U \,X(h) \ar@{=>}[r]_-{\zeta_U \ast \Id } ~&~ G_U \, X(h)
}
\end{flalign}
of natural transformations commutes.
\end{itemize}
\end{defi}

It is well-known that pseudo-functors, pseudo-natural transformations 
and modifications form a $2$-category, see e.g.\ 
\cite[Appendix A.1]{SchommerPries} and \cite[Chapter 3]{Fiore}.
Selecting only those pseudo-functors that satisfy descent 
leads to the $2$-category defined below.
\begin{defi}\label{def:stack2cat}
We denote by $\St$ the {\em $2$-category of stacks} (of categories). Its objects are 
stacks (see Definition \ref{def:stack}), $1$-morphisms are stack morphisms 
(see Definition \ref{def:stackmorphism}) and $2$-morphisms are given in Definition \ref{def:stack2morphism}. 
\end{defi}

We conclude this section by recalling briefly some important constructions involving
stacks that will be needed in the bulk of our paper.
\begin{description}
\item[$2$-Yoneda Lemma:] There exists a $2$-functor (called {\em $2$-Yoneda embedding})
\begin{flalign}\label{eqn:2Yoneda}
\und{(-)}\,:\,\Man~\longrightarrow~\St
\end{flalign}
from the category of manifolds to the $2$-category of stacks. It assigns
to a manifold $N\in\Man$ the stack $\und{N} := \Man(-,N) :\Man^\op \to \Cat\,,~U\mapsto
\Man(U,N)$, where the set $\Man(U,N) = C^\infty(U,N)$ of smooth maps
is regarded as a category with only identity morphisms.
This $2$-functor is fully faithful, i.e.\ manifolds can be equivalently regarded as stacks.
Even more, for every $U\in\Man$ and $X\in\St$, there exists a natural equivalence
\begin{flalign}
\St(\und{U},X)\,\simeq\,X(U)
\end{flalign}
between the category of morphisms from $\und{U}$ to $X$ and the category
obtained by evaluating $X$ on $U$. As a consequence,
the category $X(U)$ admits a useful interpretation as the category of  ``smooth maps''
$\und{U}\to X$ from the manifold $U$ to the stack $X$. In particular, 
for $U = \{\ast\}$ the point, we can interpret $X(\{\ast\})$ as the 
category of ``global points'' $\und{\{\ast\}} \to X$ and similarly, 
for $U=\bbR$ the line, we can interpret $X(\bbR)$ as the category 
of ``smooth curves'' $\und{\bbR}\to X$ in the stack $X$. 

\item[Products of stacks:] Given any two stacks $X,Y\in \St$, one defines 
the {\em product stack} $X\times Y \in \St$ in terms of the pseudo-functor
\begin{flalign}
\nn X\times Y \,:\, \Man^\op~&\longrightarrow~\Cat\quad,\\
\nn U ~&\longmapsto~X(U)\times Y(U)\quad,\\
\big(h:U\to U^\prime\big) ~&\longmapsto~\big(X(h)\times Y(h) : X(U^\prime)\times Y(U^\prime) \to
X(U)\times Y(U) \big)\quad,\label{eqn:productstack}
\end{flalign}
together with the obvious coherence isomorphisms induced from $X$ and $Y$.
For the particular case of two manifolds $M,N\in\Man$, one finds that the product stack
$\und{M}\times\und{N} \simeq \und{M\times N}$ is equivalent to the stack associated with the product manifold.

\item[Mapping stacks:] Given any two stacks $X,Y\in \St$, one defines
the {\em mapping stack} $\Map(X,Y)\in \St$ in terms of
the (strict) $2$-functor  
\begin{flalign}
\nn \Map(X,Y)\,:\, \Man^\op ~&\longrightarrow~ \Cat\quad,\\
\nn U~&\longmapsto~\St(X\times \und{U},Y)\quad,\\
\big(h:U\to U^\prime\big)~&\longmapsto~ \big( (\id\times \und{h})^\ast : \St(X\times \und{U^\prime},Y)\to \St(X\times \und{U},Y)\big)
\quad,\label{eqn:mappingstack}
\end{flalign}
where $ (\id\times \und{h})^\ast := (-)\circ (\id\times \und{h})$ denotes pre-composition.
For the particular case of two manifolds $M,N\in\Man$, one finds that the
mapping stack $\Map(\und{M},\und{N})$ is equivalent to the usual functor of points
for the mapping space of manifolds. See e.g.\ \cite[Section 3.2]{BSreview}
and \cite{BSCahiers} for more details on the latter.
\end{description}

%%%%%%%%%%%%%%%%%%%%%%%%%%%%%%%%%%%%%%%%%%%%%%%%

\section{\label{sec:definition}Smooth $1$-dimensional AQFTs}
An $m$-dimensional algebraic quantum field theory (AQFT) \cite{BFV,FewsterVerch} is a functor
$\AAA : \Loc_m\to \astAlg$ from the category $\Loc_m$ of $m$-dimensional 
(globally hyperbolic) Lorentzian spacetimes to the category $\astAlg$ of associative and unital $\ast$-algebras 
over $\bbC$. This functor is required to satisfy certain physically motivated axioms, most notably the Einstein causality axiom
expressing that every two causally disjoint observables must commute with each other.
Such structures can be described most effectively in terms of operad theory
and one observes that the category $\AQFT_m$ of $m$-dimensional AQFTs
is the category of algebras over a suitable colored operad, see \cite{BSWoperad,BSWinvolutive}
for the details.
The case of $m=1$ dimensions, which physically represents AQFTs on time intervals 
(i.e.\ quantum mechanics), is structurally much simpler because causal disjointness, and hence the associated 
Einstein causality axiom, is a phenomenon arising only in dimension $m\geq 2$. As a consequence,
the category $\AQFT_1 = \Fun(\Loc_1,\astAlg)$ of $1$-dimensional AQFTs is simply a functor category.
\sk

The aim of this section is to introduce a smooth refinement of $1$-dimensional AQFTs.
This means that we will upgrade the categories $\Loc_1$ and $\astAlg$ to stacks
of categories, which encode suitable concepts of smoothly $U$-parametrized families of spacetimes 
and algebras, for all manifolds $U\in\Man$. A smooth $1$-dimensional AQFT
will then be defined as a stack morphism between these two stacks, which in particular
means that smooth AQFTs map smooth $U$-families of spacetimes
to smooth $U$-families of algebras. Loosely speaking,
one may say that ``smooth AQFTs respond smoothly to smooth variations of spacetimes''.
\sk

Our approach to smooth AQFTs introduces also a further layer of smoothness,
namely we can define a stack $\AQFT^\infty_1\in\St$ of smooth $1$-dimensional AQFTs.
Through this stack we obtain a  natural concept of smoothly $U$-parametrized families of smooth AQFTs,
which for the special case $U=\bbR$ leads to a notion of smooth curves of smooth AQFTs.
We shall illustrate later in Section \ref{sec:examples} that smooth variations of the 
external parameters of a theory, such as the mass parameter, gives rise to such smooth families.
\sk

Throughout the whole paper we restrict our attention to the simplest case given by
$1$-dimensional AQFTs. We expect that a generalization to higher-dimensional AQFTs 
is possible by using similar techniques, however there are certain additional
technical difficulties and challenges that we explain in more detail in Section \ref{sec:higher}.

\subsection{\label{subsec:Locstack}The stack $\Loc_1^\infty$}
In this subsection we introduce a stack $\Loc_1^\infty : \Man^\op\to\Cat$ 
that provides a smooth refinement of the category $\Loc_1$
of $1$-dimensional spacetimes. Let us first recall how the latter
category is defined. A $1$-dimensional spacetime (i.e.\ time interval)
may be described in terms of a pair $(I,e)$, where $I\subseteq \bbR$
is an open interval and $e\in\Omega^1(I)$ is a non-degenerate
$1$-form that encodes the geometry and orientation of $I$.
(In physics terminology, one may call $e$ a {\em $1$-bein}, i.e.\ a $1$-dimensional vielbein.) 
A morphism $f: (I,e)\to (I^\prime,e^\prime)$ in $\Loc_1$ 
is an open embedding $f : I\to I^\prime$ of intervals that preserves the 
$1$-forms, i.e.\ $f^\ast (e^\prime) =e$.
\sk

Given any manifold $U\in \Man$, the category $\Loc_1^\infty(U)$
is supposed to describe smooth $U$-families
of $1$-dimensional spacetimes. A suitable way to formalize those
is through fiber bundles and their vertical geometry.
\begin{defi}\label{def:familyLocobjects}
A {\em smooth $U$-family of $1$-dimensional spacetimes} 
is a pair $(\pi:M\to U,E)$ consisting of a (locally trivializable) 
fiber bundle $\pi : M\to U$ with typical fiber an open interval $I\subseteq \bbR$
and a non-degenerate vertical $1$-form $E\in \Omega^1_\ver(M)$ on the total space.
\end{defi}
\begin{rem}\label{rem:interpretationLocobjects}
The interpretation of this definition is as follows:
Given any pair $(\pi : M\to U,E)$ as in Definition \ref{def:familyLocobjects}, 
one obtains, for every point $x\in U$, a $1$-dimensional spacetime $(M\vert_x,E\vert_x):=
(\pi^{-1}(\{x\}), E\vert_{\pi^{-1}(\{x\})}) \in \Loc_1$ by restricting
to the fiber over $x\in U$. Due to the smooth fiber bundle structure,
it makes sense to interpret this pair as depending smoothly on $x\in U$.
\end{rem}

A natural concept of morphisms $f: (\pi:M\to U,E)\to (\pi^\prime: M^\prime\to U,E^\prime)$
between smooth $U$-families of $1$-dimensional spacetimes 
is given by fiber bundle maps
\begin{flalign}\label{eqn:fiberbundlemap}
\xymatrix@C=1em{
\ar[dr]_-{\pi}M\ar[rr]^-{f} && M^\prime\ar[dl]^-{\pi^\prime}\\
&U&
}
\end{flalign}
that preserve the $1$-forms, i.e.\ $f^\ast (E^\prime) = E$,
and that are in a suitable sense ``open embeddings of fiber bundles''.
Indeed, from the AQFT point of view, it is quite natural to consider open 
embeddings as they allow to push forward compactly supported sections 
of vector bundles, which is crucial to construct examples.
We would like to emphasize that there exist a priori
different concepts of what open embeddings of fiber bundles could be.
For example, we could demand the {\em point-wise} condition 
that the restriction $f\vert_x : M\vert_x\to M^\prime\vert_x$
to the fiber over every point $x\in U$ is an open embedding of manifolds.
Unfortunately, this simple point-wise condition is incompatible 
with pushing forward vertically compactly supported functions 
on the total spaces. Hence, the correct
concept of ``open embeddings of fiber bundles'' should be in some sense
more uniform on $U$. There exist a priori different options to formalize this,
but fortunately the three main candidates are equivalent.
\begin{lem}\label{lem:Ulocalembedding}
Let $f :(\pi:M\to U)\to (\pi^\prime:M^\prime \to U)$ be a fiber bundle map. 
Then the following three statements are equivalent:
\begin{enumerate}
\item For each $x\in U$, there exists an open neighborhood $U_x \subseteq U$, such that
the restriction $f\vert_{U_x} : M\vert_{U_x} \to M^\prime\vert_{U_x}$ is an open embedding of manifolds.

\item The map of total spaces $f: M \to M^\prime$ is an open embedding of manifolds.

\item For each open subset $\tilde{U}\subseteq U$, the restriction 
$f\vert_{\tilde{U}} : M\vert_{\tilde{U}} \to M^\prime\vert_{\tilde{U}}$ is an open embedding of manifolds.
\end{enumerate}
\end{lem}
\begin{proof}
1.~$\Rightarrow$ 2.: From the hypothesis it is clear that $f:M\to f(M)$ is a bijection of sets. 
Furthermore, for each $x\in U$,  there exists an open neighborhood $U_x\subseteq U$
such that $f\vert_{U_x} : M\vert_{U_x} \to M^\prime\vert_{U_x}$ is an open embedding,
which implies that $f:M\to f(M)$ is a diffeomorphism. To show that the image $f(M)\subseteq M^\prime$ is open,
observe that $f\vert_{U_x}(M\vert_{U_x})\subseteq M^\prime \vert_{U_x}$
is by hypothesis open and that $M^\prime \vert_{U_x}\subseteq M^\prime$ is open too.
Hence, $f(M)=\bigcup_{x\in U} f\vert_{U_x}(M\vert_{U_x})\subseteq M^\prime$ is open.
\sk

2.~$\Rightarrow$ 3.: The open embedding $f:M\to M^\prime$ factors as a diffeomorphism 
$f : M \to f(M)$ followed by an open inclusion $f(M)\subseteq M^\prime$. Take any open subset
$\tilde{U}\subseteq U$ and consider the restriction $f\vert_{\tilde{U}} : M\vert_{\tilde{U}}\to M^\prime\vert_{\tilde{U}}$,
which factors as a map $f\vert_{\tilde{U}} : M\vert_{\tilde{U}}\to f(M\vert_{\tilde{U}})$ followed
by an inclusion $f(M\vert_{\tilde{U}})\subseteq M^\prime\vert_{\tilde{U}}$.
Because $f$ is a fiber bundle map, we have that  $f(M\vert_{\tilde{U}}) = f(M)\cap M^\prime\vert_{\tilde{U}}$, 
which implies that $f\vert_{\tilde{U}} : M\vert_{\tilde{U}}\to f(M\vert_{\tilde{U}})$ is a diffeomorphism
and that $f(M\vert_{\tilde{U}})\subseteq M^\prime\vert_{\tilde{U}}$ is an open inclusion. Hence, 
$f\vert_{\tilde{U}} : M\vert_{\tilde{U}}\to M^\prime\vert_{\tilde{U}}$ is an open embedding.
\sk

3.~$\Rightarrow$ 1.: Trivial.
\end{proof}

\begin{defi}\label{def:Loc_1(U)}
For any manifold $U\in \Man$, we denote by $\Loc_1^\infty(U)$ the category
whose objects are smooth $U$-families of $1$-dimensional spacetimes
$(\pi:M\to U,E)$, see Definition \ref{def:familyLocobjects}, 
and whose morphisms $f: (\pi:M\to U, E)\to (\pi^\prime : M^\prime\to U,E^\prime)$
are fiber bundle maps \eqref{eqn:fiberbundlemap}
that preserve the $1$-forms, i.e.\ $f^\ast (E^\prime) = E$, and that satisfy
any of the three equivalent conditions in Lemma \ref{lem:Ulocalembedding}.
\end{defi}

In order to define a (pre)stack $\Loc_1^\infty : \Man^\op\to \Cat$ (see Definition \ref{def:prestack}), 
we have to assign to every morphism $h:U\to U^\prime$ in $\Man$ a functor
\begin{flalign}\label{eqn:Locfunctor1}
h^\ast\,:=\, \Loc_1^\infty(h) \,:\, \Loc_1^\infty(U^\prime)~\longrightarrow~\Loc_1^\infty(U)\quad. 
\end{flalign}
Let us recall that, given any fiber bundle $\pi: M\to U^\prime$,
one may form the {\em pullback bundle}
\begin{flalign}\label{eqn:pullbackbundle}
\xymatrix{
\ar@{-->}[d]_-{\pi_h}h^\ast M \ar@{-->}[r]^-{\overline{h}^M} & M\ar[d]^-{\pi}\\
U\ar[r]_-{h}&U^\prime
}
\end{flalign}
which is a locally trivializable fiber bundle with the same typical fiber as $\pi:M\to U^\prime$.
We then define the functor \eqref{eqn:Locfunctor1} on objects as
\begin{subequations}\label{eqn:Locfunctor2}
\begin{flalign}\label{eqn:Locfunctor2a}
h^\ast\big(\pi :M\to U^\prime,E\big) \,:=\,\big(\pi_h : h^\ast M \to U, \overline{h}^{M\ast} (E)\big)
\end{flalign}
and on morphisms $f : (\pi:M\to U^\prime, E)\to (\pi^\prime : M^\prime\to U^\prime,E^\prime)$ as
\begin{flalign}\label{eqn:Locfunctor2b}
h^\ast f \,:\,   \big(\pi_h: h^\ast M\to U, \overline{h}^{M\ast} (E)\big)~\longrightarrow~\big(\pi^\prime_h : h^\ast M^\prime\to U,\overline{h}^{M^\prime\ast} (E^\prime)\big)\quad,
\end{flalign}
where the fiber bundle map $h^\ast f $ is defined uniquely through the universal 
property of pullback bundles by the commutative diagram
\begin{flalign}\label{eqn:Locfunctor2c}
\xymatrix@C=1em@R=2em{
~&~h^\ast M^\prime\ar[rr]^-{\overline{h}^{M^\prime}} \ar@{-}[d]_(.7){\pi^\prime_h}~&~ ~&~M^\prime \ar[dd]^-{\pi^\prime}\\
\ar[dr]_-{\pi_h}\ar@{-->}[ur]^-{h^\ast f} h^\ast M \ar[rr]^(.7){\overline{h}^{M}}~&~\ar[d] ~&~ M \ar[dr]_-{\pi}\ar[ru]^(.4){f}~&~ \\
~&~ U \ar[rr]_-{h}~&~ ~&~ U^\prime
}
\end{flalign}
\end{subequations}
The fact that $h^\ast f$ preserves the $1$-forms, i.e.\ $(h^\ast f)^\ast \overline{h}^{M^\prime\ast} (E^\prime)
= \overline{h}^{M\ast} (E)$, is a direct consequence of this diagram and $f^\ast (E^\prime) = E$.
Furthermore, the condition 1.~of Lemma \ref{lem:Ulocalembedding} for $h^\ast f$ can be easily proven
using that $f$ satisfies this condition. Hence, \eqref{eqn:Locfunctor2b} defines 
a morphism in $\Loc_1^\infty(U)$.
\begin{propo}\label{prop:Locstack}
The prestack $\Loc_1^\infty : \Man^\op\to \Cat$ defined by Definition \ref{def:Loc_1(U)}, \eqref{eqn:Locfunctor1},
\eqref{eqn:Locfunctor2} and the canonical coherence isomorphisms given by the universal property of pullback bundles
is a stack, i.e.\ it satisfies the descent condition from Definition \ref{def:stack}.
\end{propo}
\begin{proof}
This is a direct consequence of descent for fiber bundles and differential forms
and of the fact that the first condition on the fiber bundle morphisms stated in Lemma 
\ref{lem:Ulocalembedding} is a local condition on $U\in\Man$. 
In more detail, spelling out descent for objects $(\pi:M\to U, E) \in \Loc_1^\infty(U)$, 
one observes that it involves descent for the underlying 
fiber bundles $\pi:M\to U$ and also for the underlying $1$-forms $E$, 
which are straightforward consequences of descent for fiber bundles 
and differential forms. 
Similarly, descent for $\Loc_1^\infty(U)$-morphisms 
$f: (\pi:M\to U, E)\to (\pi^\prime : M^\prime\to U,E^\prime)$ 
involves descent for the underlying fiber bundle maps 
$f: (\pi:M\to U)\to (\pi^\prime : M^\prime\to U)$ 
and the verification that the $1$-forms are preserved and that any one 
of the equivalent conditions from Lemma \ref{lem:Ulocalembedding} holds, 
which are again both consequences of descent for fiber bundles 
and differential forms and the fact 
that the descent data fulfill these properties. 
\end{proof}
\begin{rem}
Observe that the category $\Loc_1^\infty(\{\ast\})$ of global points $\und{\{\ast\}} \to \Loc_1^\infty$
of the stack $\Loc_1^\infty$ is the ordinary category $\Loc_1$ of $1$-dimensional spacetimes.
\end{rem}

\subsection{\label{subsec:Algstack}The stack $\astAlg^\infty$}
The aim of this subsection is to develop a stack $\astAlg^\infty$ that provides
a smooth refinement of the usual category $\astAlg$ of associative and unital $\ast$-algebras
over $\bbC$. Let us recall that the latter category may be defined as the category
$\astMon_{\mathrm{rev}}(\Vec_\bbC)$ of order-reversing $\ast$-monoids in the {\em involutive} symmetric monoidal
category $\Vec_\bbC$ of complex vector spaces, see e.g.\ \cite{Jacobs,BSWinvolutive} 
for the relevant background on involutive category theory.
Our strategy is to introduce first a stack (of involutive symmetric monoidal categories)
that refines the category $\Vec_\bbC$ of vector spaces over $\bbC$
and then discuss how to form order-reversing $\ast$-monoids at the level of stacks.
\sk

Let us consider for now the case where $\bbK$ is either $\bbR$ or $\bbC$.
As a first attempt to introduce a smooth refinement of the category $\Vec_\bbK$,
we could consider the stack $\VecBun_\bbK : \Man^\op\to \Cat$ of
$\bbK$-vector bundles introduced in Example \ref{ex:vecbun}.
To a manifold $U\in\Man$, this stack assigns the category
$\VecBun_\bbK(U)$ of (locally trivializable and finite rank) $\bbK$-vector bundles
over $U$. Considering as in Remark \ref{rem:interpretationLocobjects} the fibers over points $x\in U$, 
every vector bundle can be interpreted as a smooth $U$-family of $\bbK$-vector spaces. 
The problem with this first attempt is that the fibers of vector bundles 
are (by definition) {\em finite-dimensional} vector spaces,
while examples of AQFTs, even in dimension $1$, require infinite-dimensional vector spaces,
such as the vector spaces underlying the canonical commutation relation algebras.
A natural way to enlarge the category $\VecBun_{\bbK}(U)$ in order to capture such
infinite-dimensional aspects is to pass (via the sheaf of sections functor) 
to the category $\Sh_{C^\infty_\bbK}(U)$ of sheaves of $C^\infty_{\bbK,U}$-modules over $U\in\Man$. 
Here $C^\infty_{\bbK,U} : \Open(U)^\op\to \Alg_\bbK\,,~(\tilde{U}\subseteq U)\mapsto C_\bbK^\infty(\tilde{U}) $ 
denotes the sheaf of $\bbK$-valued smooth functions on $U$. 
Indeed, $\VecBun_\bbK(U)$ embeds fully faithfully in 
$\Sh_{C^\infty_\bbK}(U)$ and the essential image 
consists of locally free $C^\infty_{\bbK,U}$-modules of finite rank, 
see e.g.\ \cite[Chapter 2]{Ramanan}. 
\begin{rem}
We would like to note that there are also alternative candidates 
to enlarge the category $\VecBun_\bbK(U)$ to include 
such infinite-dimensional aspects. For example, one could imagine to work
with bundles over $U$ whose fibers are e.g.\ locally convex,
bornological or diffeological vector spaces. However, to make
this a valid choice, one would have to 
confirm that such categories assemble into a stack, as it is
the case for the sheaf categories $U\mapsto \Sh_{C^\infty_\bbK}(U)$, 
see Proposition \ref{prop:Shmodulestack} below. 
As another alternative, one could search directly for a stack providing a
smooth refinement of the category of $C^\ast$-algebras. 
To the best of our knowledge, such a stack has not yet been studied, 
but we believe that this may be related to the concept of
continuous bundles/fields of $C^\ast$-algebras,
see e.g\ \cite{CastBundles} and \cite[Section 10.3]{Dixmier}.
\end{rem}

Following \cite{Kashiwara}, let us now describe the (pre)stack 
$\Sh_{C_\bbK^\infty} : \Man^\op\to \Cat$ of sheaves of $C^\infty_\bbK$-modules in more detail.
To each manifold $U\in\Man$,  it assigns the category
$\Sh_{C^\infty_\bbK}(U)$ of sheaves of $C^\infty_{\bbK,U}$-modules over $U\in\Man$, 
with morphisms given by $C^\infty_{\bbK,U}$-linear sheaf morphisms.
To a morphism $h : U\to U^\prime$ in $\Man$, it assigns the functor
\begin{subequations}\label{eqn:sheafmodstack}
\begin{flalign}\label{eqn:sheafmodstack1}
h^\ast\,:=\, \Sh_{C^\infty_\bbK}(h)\,:\,\Sh_{C^\infty_\bbK}(U^\prime)~\longrightarrow~\Sh_{C^\infty_\bbK}(U)
\end{flalign}
that acts on $V\in\Sh_{C^\infty_\bbK}(U^\prime) $ as
\begin{flalign}\label{eqn:sheafmodstack2}
h^\ast V \,:=\, h^{-1}(V) \otimes_{h^{-1}(C^\infty_{\bbK,U^\prime})} C^\infty_{\bbK,U}\quad,
\end{flalign}
\end{subequations}
where $h^{-1}$ is the inverse image sheaf functor and $\otimes_{h^{-1}(C^\infty_{\bbK,U^\prime})} $
denotes the relative tensor product of sheaves of modules. 
Together with the canonical coherence isomorphisms associated with relative tensor products
and inverse image functors, this defines a prestack
$\Sh_{C^\infty_\bbK}:\Man^\op\to \Cat$ in the sense of Definition \ref{def:prestack}.
The following result is well-known, see e.g.\ \cite[Proposition 19.4.7]{Kashiwara}.
\begin{propo}\label{prop:Shmodulestack}
For $\bbK$ being either $\bbR$ or $\bbC$, the prestack $\Sh_{C^\infty_\bbK}: \Man^\op\to \Cat$ 
defined above is a stack, i.e.\ it satisfies
the descent condition from Definition \ref{def:stack}.
\end{propo}
\begin{rem}
Observe that the category $\Sh_{C^\infty_\bbK}(\{\ast\})$ of global points 
$\und{\{\ast\}} \to\Sh_{C^\infty_\bbK}$
of the stack $\Sh_{C^\infty_\bbK}$ is the ordinary category $\Vec_\bbK$ of vector spaces over $\bbK$.
\end{rem}

As explained at the beginning of this subsection, we interpret 
the stack $\Sh_{C^\infty_\bbK}$ as a smooth refinement of
the category $\Vec_\bbK$ of vector spaces over $\bbK$. 
In order to introduce a smooth refinement of the 
category $\astAlg= \astMon_{\mathrm{rev}}(\Vec_\bbC)$ of
associative and unital $\ast$-algebras over $\bbC$, we have to define
an involutive symmetric monoidal structure on $\Sh_{C^\infty_\bbC}$.
To achieve this goal, let us first observe that, for both $\bbK=\bbR$ or $\bbC$,
the category $\Sh_{C^\infty_\bbK}(U)$ of sheaves of $C^\infty_{\bbK,U}$-modules
over each $U\in\Man$ is symmetric monoidal
with respect to the relative tensor product
\begin{flalign}
V\otimes_{C^\infty_{\bbK,U}} V^\prime \,\in\, \Sh_{C^\infty_\bbK}(U)\quad,
\end{flalign}
for all $V,V^\prime\in \Sh_{C^\infty_\bbK}(U)$. (The monoidal unit
is $C^\infty_{\bbK,U}\in  \Sh_{C^\infty_\bbK}(U)$, regarded as a sheaf of $C^\infty_{\bbK,U}$-modules.)
Furthermore, for each morphism $h: U\to U^\prime$ in $\Man$, the
functor \eqref{eqn:sheafmodstack} is strong
symmetric monoidal via the coherence isomorphisms
\begin{subequations}\label{eqn:ShSMcoherences}
\begin{flalign}
\nn h^\ast\big(V\otimes_{C^\infty_{\bbK,U^\prime}} V^\prime\big)
&= h^{-1}\big(V\otimes_{C^\infty_{\bbK,U^\prime}} V^\prime\big) \otimes_{h^{-1}(C^\infty_{\bbK,U^\prime})}C^\infty_{\bbK,U}\\
\nn &\cong h^{-1}(V)\otimes_{h^{-1}(C^\infty_{\bbK,U^\prime})} h^{-1}(V^\prime)\otimes_{h^{-1}(C^\infty_{\bbK,U^\prime})} C^\infty_{\bbK,U}\\
\nn &\cong \big(h^{-1}(V)\otimes_{h^{-1}(C^\infty_{\bbK,U^\prime})}  C^\infty_{\bbK,U}\big)\otimes_{C^\infty_{\bbK,U}} 
\big(h^{-1}(V^\prime)\otimes_{h^{-1}(C^\infty_{\bbK,U^\prime})} C^\infty_{\bbK,U}\big)\\
&=(h^\ast V)\otimes_{C^\infty_{\bbK,U}} (h^\ast V^\prime)
\end{flalign}
and
\begin{flalign}
h^\ast C^\infty_{\bbK,U^\prime} = h^{-1}(C^\infty_{\bbK,U^\prime})\otimes_{h^{-1}(C^\infty_{\bbK,U^\prime})} C^\infty_{\bbK,U} 
\cong C^\infty_{\bbK,U}\quad.
\end{flalign}
\end{subequations}
One can check that the canonical coherence isomorphisms of the stack 
$\Sh_{C^\infty_\bbK}$ are monoidal natural transformations.
\begin{cor}
The stack $\Sh_{C^\infty_\bbK}$ in Proposition \ref{prop:Shmodulestack}
is canonically a stack $\Sh_{C^\infty_\bbK} : \Man^\op\to \SMCat$ with
values in the $2$-category $\SMCat$ 
of symmetric monoidal categories, strong symmetric monoidal 
functors and monoidal natural transformations.
\end{cor}

In the case of $\bbK=\bbC$, we can define further, for each $U\in\Man$,
an involution endofunctor $\overline{(-)} : \Sh_{C^\infty_\bbC}(U)\to \Sh_{C^\infty_\bbC}(U)$.
It assigns to an object $V\in  \Sh_{C^\infty_\bbC}(U)$ the complex 
conjugate sheaf of $C^\infty_{\bbC,U}$-modules $\overline{V} \in  \Sh_{C^\infty_\bbC}(U)$
which, as a sheaf, coincides with $V$, but the $C^\infty_{\bbC,U}$-module
structure is defined via complex conjugation of $\bbC$-valued functions as
$\overline{v}\cdot a :=\overline{v\cdot a^\ast}$, for all $\overline{v}\in\overline{V}$
and $a\in C^\infty_{\bbC,U}$. Clearly, the endofunctor $\overline{(-)}$ squares to the identity and hence 
defines an involutive structure on the category $\Sh_{C^\infty_\bbC}(U)$, see \cite{Jacobs,BSWinvolutive}.
Observe that $\overline{(-)}$ is canonically a strong symmetric monoidal functor
with respect to the symmetric monoidal structure on $\Sh_{C^\infty_\bbC}(U)$ 
introduced above. Hence, we obtain that $\Sh_{C^\infty_\bbC}(U)$ is
an involutive symmetric monoidal category, for every $U\in\Man$.
Furthermore, for each morphism $h: U\to U^\prime$ in $\Man$, the
symmetric monoidal functor \eqref{eqn:sheafmodstack} is involutive via the
coherence isomorphisms
\begin{flalign}
\nn h^\ast \overline{V} &= h^{-1}(\overline{V})\otimes_{h^{-1}(C^\infty_{\bbC,U^\prime})} C^\infty_{\bbC,U}
\cong \overline{h^{-1}(V)} \otimes_{h^{-1}(C^\infty_{\bbC,U^\prime})} \overline{C^\infty_{\bbC,U}}\\
&\cong \overline{h^{-1}(V)\otimes_{h^{-1}(C^\infty_{\bbC,U^\prime})} C^\infty_{\bbC,U} }= \overline{ h^\ast V }\quad,
\label{eqn:ShInvcoherences}
\end{flalign}
for all $V\in \Sh_{C^\infty_\bbC}(U^\prime)$, where in the second step we used complex conjugation
$\ast : C^\infty_{\bbC,U} \to \overline{C^\infty_{\bbC,U}}$.
Summing up, we obtain
\begin{cor}\label{cor:ShstackISMCAT}
For $\bbK=\bbC$, the stack $\Sh_{C^\infty_\bbC}$ in Proposition \ref{prop:Shmodulestack}
is canonically a stack $\Sh_{C^\infty_\bbC} : \Man^\op\to \mathbf{ISMCat}$ with
values in the $2$-category $\mathbf{ISMCat}$ 
of involutive symmetric monoidal categories, involutive strong symmetric monoidal 
functors and involutive monoidal natural transformations.
\end{cor}

With these preparations, it is now straightforward to introduce a (pre)stack
$\astAlg^\infty$ that provides a smooth refinement of the ordinary category
$\astAlg=\astMon_{\mathrm{rev}}(\Vec_\bbC)$ of associative and unital $\ast$-algebras over $\bbC$.
Using that forming order-reversing $\ast$-monoids
is a $2$-functor $\astMon_{\mathrm{rev}} : \mathbf{ISMCat}\to \Cat$, see
\cite{Jacobs,BSWinvolutive}, we define a prestack (in the sense of Definition \ref{def:prestack})
by the composition
\begin{flalign}\label{eqn:astAlgstack}
\astAlg^\infty\,:=\, \astMon_{\mathrm{rev}}\circ \Sh_{C^\infty_\bbC} \,:\,\Man^\op ~\longrightarrow~\Cat\quad.
\end{flalign}
More explicitly, this prestack assigns, to each manifold $U\in\Man$, the category
$\astAlg^\infty(U) =\astMon_{\mathrm{rev}}(\Sh_{C^\infty_\bbC}(U))$ of 
order-reversing $\ast$-monoids in the involutive
symmetric monoidal category $\Sh_{C^\infty_\bbC}(U)$. An object in this category
is a quadruple $(A,\mu,\eta,\ast)$, where $A\in \Sh_{C^\infty_\bbC}(U)$
is a sheaf of $C^\infty_{\bbC,U}$-modules on $U\in\Man$ and
\begin{flalign}
\mu \,:\, A\otimes_{C^\infty_{\bbC,U}} A ~\longrightarrow~A\quad,\qquad
\eta \,:\, C^\infty_{\bbC,U} ~\longrightarrow~A\quad,\qquad
\ast \,:\, A~\longrightarrow~ \overline{A}
\end{flalign}
are morphisms in $\Sh_{C^\infty_\bbC}(U)$ that satisfy the axioms of an associative and unital $\ast$-algebra.
A morphism $\kappa : (A,\mu,\eta,\ast)\to (A^\prime,\mu^\prime,\eta^\prime,\ast^\prime)$
in $\astAlg^\infty(U)$ is a morphism $\kappa : A\to A^\prime$ in $\Sh_{C^\infty_\bbC}(U)$
that preserves the multiplications, units and involutions. To each morphism 
$h:U\to U^\prime$ in $\Man$, the prestack $\astAlg^\infty$ assigns the functor
\begin{flalign}
h^\ast\,:=\, \astAlg^\infty(h)\,:\, \astAlg^\infty(U^\prime)~\longrightarrow~\astAlg^\infty(U)
\end{flalign}
that maps $(A,\mu,\eta,\ast)\in \astAlg^\infty(U^\prime)$ to the object $h^\ast A \in\Sh_{C^\infty_\bbC}(U) $ 
given in \eqref{eqn:sheafmodstack2}, endowed with the structure maps
\begin{subequations}
\begin{flalign}
&\xymatrix{
(h^\ast A)\otimes_{C^\infty_{\bbC,U}} (h^\ast A)\, \cong\, h^\ast\big(A\otimes_{C^\infty_{\bbC,U^\prime}}A \big) \ar[r]^-{h^\ast \mu}~&~ h^\ast A
}\quad,\\
&\xymatrix{
C^\infty_{\bbC,U} \,\cong\, h^\ast C^\infty_{\bbC,U^\prime} \ar[r]^-{h^\ast \eta } ~&~ h^\ast A
}\quad,\\
&\xymatrix{
h^\ast A \ar[r]^-{h^\ast \ast }~&~h^\ast \overline{A}  \,\cong\, \overline{h^\ast A}\quad,
}
\end{flalign}
\end{subequations}
obtained by using the coherence isomorphisms of the involutive symmetric monoidal 
stack $\Sh_{C^\infty_\bbC}$ from Corollary \ref{cor:ShstackISMCAT}.
\begin{propo}\label{propo:Algstack}
The prestack $\astAlg^\infty$ defined in \eqref{eqn:astAlgstack} is a stack,
i.e.\ it satisfies the descent condition from Definition \ref{def:stack}.
\end{propo}
\begin{proof}
Let $\{U_\alpha\subseteq U\}$ be any open cover of any $U\in \Man$.
The key step is to realize that the descent category $\astAlg^\infty(\{U_\alpha\subseteq U\})$
coincides with the category $\astMon_{\mathrm{rev}}\big(\Sh_{C^\infty_\bbC}(\{U_\alpha\subseteq U\})\big)$
of order-reversing $\ast$-monoids in the descent category $\Sh_{C^\infty_\bbC}(\{U_\alpha\subseteq U\})$,
which we endow with the  involutive symmetric monoidal structure
given by
\begin{subequations}\label{eqn:descentISM}
\begin{flalign}
\big(\{V_\alpha\},\{\varphi_{\alpha\beta}\}\big)\otimes\big(\{V^\prime_\alpha\},\{\varphi^\prime_{\alpha\beta}\}\big)
\,&:=\,\big(\{V_\alpha\otimes_{C^\infty_{\bbC,U_\alpha}} V^\prime_\alpha\} ,
\{\varphi_{\alpha\beta}\otimes_{C^\infty_{\bbC,U_{\alpha\beta}}}\varphi_{\alpha\beta}^\prime\} \big)
\\
\overline{ \big(\{V_\alpha\},\{\varphi_{\alpha\beta}\}\big)} \,&:=\, \big(\{\overline{V_\alpha}\},\{\overline{\varphi_{\alpha\beta}}\}\big)\quad,
\end{flalign}
\end{subequations}
where we have suppressed the coherence isomorphisms \eqref{eqn:ShSMcoherences}
and \eqref{eqn:ShInvcoherences}. Fully explicitly, the conjugated cocycle 
$\overline{\varphi_{\alpha\beta}}$ is given by
\begin{flalign}
\xymatrix{
\overline{V_\beta}\vert_{U_{\alpha\beta}} \cong \overline{V_\beta\vert_{U_{\alpha\beta}}} 
\ar[r]^-{\overline{\varphi_{\alpha\beta}}} ~&~ \overline{V_\alpha\vert_{U_{\alpha\beta}}} \cong
\overline{V_\alpha}\vert_{U_{\alpha\beta}}
} \quad,
\end{flalign} 
and similarly for the tensor product cocycle 
$\varphi_{\alpha\beta}\otimes_{C^\infty_{\bbC,U_{\alpha\beta}}}\varphi_{\alpha\beta}^\prime$.
The functor to the descent category $\Sh_{C^\infty_\bbC}(U) \to \Sh_{C^\infty_\bbC}(\{U_\alpha\subseteq U\})$
given in \eqref{eqn:functordescentcategory} carries a canonical involutive symmetric monoidal structure
and it is an equivalence in $\mathbf{ISMCat}$ because $\Sh_{C^\infty_\bbC}$ is a stack.
Applying the $2$-functor $\astMon_{\mathrm{rev}} : \mathbf{ISMCat}\to \Cat$ that takes order-reversing
$\ast$-monoids then yields the equivalence of categories  
$\astAlg^\infty(U)\to \astAlg^\infty(\{U_\alpha\subseteq U\})$
that proves descent for $\astAlg^\infty$.
\end{proof}
\begin{rem}
Observe that the category $\astAlg^\infty(\{\ast\})$ of global points $\und{\{\ast\}} \to\astAlg^\infty$
of the stack $\astAlg^\infty$ is the ordinary category $\astAlg$ of 
associative and unital $\ast$-algebras over $\bbC$.
\end{rem}

\subsection{\label{subsec:AQFTstack}The stack $\AQFT_1^\infty$}
With these preparations, we are now ready to introduce a natural concept of smooth $1$-dimensional AQFTs.
Recalling that the category of ordinary $1$-dimensional AQFTs is described as the functor category
$\AQFT_1 := \Fun(\Loc_1,\astAlg)$, we propose the following
\begin{defi}\label{def:AQFTstack}
The {\em stack of smooth $1$-dimensional AQFTs} is defined as the mapping stack (see \eqref{eqn:mappingstack})
\begin{flalign}
\AQFT_1^\infty\,:=\, \Map\big(\Loc_1^\infty,\astAlg^\infty\big)\,\in\,\St
\end{flalign}
from the stack $\Loc_1^\infty$ of $1$-dimensional spacetimes developed in Subsection \ref{subsec:Locstack}
to the stack $\astAlg^\infty$ of associative and unital $\ast$-algebras developed in Subsection \ref{subsec:Algstack}.
\end{defi}

This very simple definition is incredibly rich and powerful, as we shall explain throughout 
the rest of this subsection. Before discussing some of its more sophisticated consequences,
we believe that it is worth spelling out explicitly what a smooth $1$-dimensional AQFT is.
By definition, it is a global point $\und{\{\ast\}} \to \AQFT_1^\infty$ of the stack introduced in 
Definition \ref{def:AQFTstack} which, by the $2$-Yoneda Lemma, is equivalently an
object $\AAA\in \AQFT_1^\infty(\{\ast\})$. According to \eqref{eqn:mappingstack}, which defines mapping stacks,
we find that a smooth $1$-dimensional AQFT is then simply a stack morphism
$\AAA : \Loc_1^\infty \to \astAlg^\infty$. Even more explicitly, this consists of a family of functors
\begin{subequations}
\begin{flalign}
\AAA_U\,:\, \Loc_1^\infty(U)~\longrightarrow~\astAlg^\infty(U)\quad,
\end{flalign}
for all manifolds $U\in\Man$, and natural isomorphisms
\begin{flalign}
\xymatrix@C=2em@R=2em{
\ar[d]_-{h^\ast} \Loc_1^\infty(U^\prime) \ar[r]^-{\AAA_{U^\prime}}~&~ \astAlg^\infty(U^\prime)\ar[d]^-{h^\ast} \ar@{=>}[dl]_-{\AAA_h}\\
\Loc_1^\infty(U) \ar[r]_-{\AAA_{U}}~&~\astAlg^\infty(U)
}
\end{flalign}
\end{subequations}
for all morphisms $h:U\to U^\prime$ in $\Man$,
that satisfy the coherence axioms listed in Definition \ref{def:stackmorphism}.
When we interpret, as explained in the previous subsections, $\Loc_1^\infty(U)$
as the category of smooth $U$-families of $1$-dimensional spacetimes
and $\astAlg^\infty(U)$ as the category of smooth $U$-families of
algebras, the role of the functor $\AAA_U: \Loc_1^\infty(U)\to\astAlg^\infty(U)$
is to capture the response of the observable algebras to ``smooth variations of spacetimes''.
Hence, smooth AQFTs have built in a suitable concept of smooth dependence on
smooth variations of spacetimes, which we will illustrate in more detail via simple examples
in Section \ref{sec:examples}. Let us also note that the functor 
$\AAA_{\{\ast\}} : \Loc_1^\infty(\{\ast\})= \Loc_1\to\astAlg^\infty(\{\ast\})=\astAlg$
associated with the point $U=\{\ast\}$ defines an ordinary $1$-dimensional AQFT.
Hence, every smooth AQFT has an underlying ordinary AQFT
and it therefore provides a refinement of the ordinary concept.
\sk

Another interesting consequence of Definition \ref{def:AQFTstack} is
that it introduces a natural concept of  ``smooth curves of smooth AQFTs'',
or more generally of smooth $\tilde{U}$-families of smooth AQFTs, for every manifold $\tilde{U}\in\Man$.
By definition, a smooth $\tilde{U}$-family of smooth AQFTs is a $\tilde{U}$-point
$\und{\tilde{U}}\to \AQFT_1^\infty$ of the stack from Definition \ref{def:AQFTstack}
which, by the $2$-Yoneda Lemma, is equivalently an object $\BBB\in \AQFT_1^\infty(\tilde{U})$. 
From the definition of mapping stacks \eqref{eqn:mappingstack}, we obtain
that this is simply a stack morphism $ \Loc_1^\infty\times\und{\tilde{U}}\to \astAlg^\infty$,
or equivalently a stack morphism
\begin{flalign}\label{eqn:tildeUfamilyAQFT}
\BBB\,:\, \Loc_1^\infty~\longrightarrow~\Map\big(\und{\tilde{U}},\astAlg^\infty\big)
\end{flalign}
to the mapping stack from $\und{\tilde{U}}$ to $\astAlg^\infty$.
Even more explicitly, using again the $2$-Yoneda Lemma, 
this is a family of functors
\begin{subequations}
\begin{flalign}
\BBB_U\,:\, \Loc_1^\infty(U)~\longrightarrow~\astAlg^\infty(U\times\tilde{U})\quad,
\end{flalign}
for all manifolds $U\in\Man$, and natural isomorphisms
\begin{flalign}
\xymatrix@C=2em@R=2em{
\ar[d]_-{h^\ast} \Loc_1^\infty(U^\prime) \ar[r]^-{\BBB_{U^\prime}}~&~ \astAlg^\infty(U^\prime\times\tilde{U})\ar[d]^-{(h\times\id)^\ast} \ar@{=>}[dl]_-{\BBB_h}\\
\Loc_1^\infty(U) \ar[r]_-{\BBB_{U}}~&~\astAlg^\infty(U\times\tilde{U})
}
\end{flalign}
\end{subequations}
for all morphisms $h:U\to U^\prime$ in $\Man$,
that satisfy the coherence axioms listed in Definition \ref{def:stackmorphism}.
The role of the functor $\BBB_U: \Loc_1^\infty(U)\to\astAlg^\infty(U\times\tilde{U})$
is now twofold: Firstly, it captures the response of the observable algebras to ``smooth $U$-variations of spacetimes''.
Secondly, it captures the response of the observable algebras to ``smooth $\tilde{U}$-variations of
the smooth AQFT itself''. Again, this concept is best illustrated via simple examples, see
Section \ref{sec:examples}.
\sk

As another interesting consequence of Definition \ref{def:AQFTstack},
let us note that every smooth AQFT $\AAA : \und{\{\ast\}}\to\AQFT_1^\infty$ 
has a smooth automorphism group. (We refer to \cite{FewsterAut} for automorphism groups in ordinary AQFT,
which in general are not smooth groups.) This can be defined in terms of the {\em loop stack}
\begin{flalign}\label{eqn:Aut}
\xymatrix{
\ar@{-->}[d] \Aut(\AAA) \ar@{-->}[r] ~&~\und{\{\ast\}}\ar[d]^-{\AAA}\\
\und{\{\ast\}} \ar[r]_-{\AAA}~&~ \AQFT_1^\infty
}
\end{flalign}
which is a bicategorical pullback in the $2$-category $\St$ of stacks of 
categories.\footnote{The bicategorical pullback \eqref{eqn:Aut} 
exists because the $2$-category $\Cat$ admits 
all bicategorical limits, see e.g.\ \cite[Theorem 5.1]{Fiore}, 
and hence so does $\St$.}
By a direct computation of this bicategorical pullback, one finds that 
$\Aut(\AAA) : \Man^\op\to \Set\subset \Cat$ is equivalent to a sheaf of sets (i.e.\ discrete categories),
which due to the universal property of bicategorical pullbacks comes endowed 
with a group structure. This implies that $\Aut(\AAA): \Man^\op \to \mathbf{Grp}$ is a 
sheaf of groups on $\Man$, i.e.\ a smooth group from the functor of points perspective. 
\sk

Let us briefly explain how this concept of smooth automorphism groups 
is related to the more practical concept of smooth AQFTs with a smooth action 
of a Lie group. Given any Lie group $G$, we use the $2$-Yoneda embedding
to define a group object $\und{G}\in\St$ in the $2$-category of stacks
and construct the quotient stack
\begin{flalign}\label{eqn:quotientstack}
[\{\ast\}/G]\,:=\,\bicolim_{\St}\bigg(
\xymatrix@C=1em{
\und{\{\ast\}} \ar[r]~&~ \ar@<1ex>[l] \ar@<-1ex>[l] \ar@<1ex>[r] \ar@<-1ex>[r]\und{G}~&~ \ar@<2ex>[l] \ar[l] \ar@<-2ex>[l] \und{G}^2 \ar[r]\ar@<2ex>[r] \ar@<-2ex>[r] ~&~ \ar@<1ex>[l] \ar@<-1ex>[l]\ar@<3ex>[l] \ar@<-3ex>[l]\cdots
}\bigg)\,\in\,\St
\end{flalign}
associated with the trivial action of $G$ on the point $\{\ast\}$ via a bicategorical colimit.
A {\em $G$-equivariant smooth $1$-dimensional AQFT} is then defined to be a stack morphism 
\begin{flalign}
\AAA^{\mathrm{eq}} \, :\, [\{\ast\}/G]~\longrightarrow~\AQFT_1^\infty\quad.
\end{flalign}
By the universal property of the bicategorical colimit \eqref{eqn:quotientstack},
this datum is equivalent to a smooth AQFT $\AAA : \und{\{\ast\}}\to\AQFT_1^\infty$
together with a $2$-automorphism $\AAA_2$ of the stack morphism
$\und{G}\to \und{\{\ast\}}\stackrel{\AAA}{\to} \AQFT_1^\infty$ 
that satisfies certain compatibility conditions arising from the
face and degeneracy maps in \eqref{eqn:quotientstack}.\footnote{Recalling Definition \ref{def:stack2morphism},
let us also state these conditions explicitly at the level of the component
natural automorphisms ${\AAA_{2}}_U$ of the functors 
$\und{G}(U)\to \{\ast\} \stackrel{\AAA_U}{\to}\AQFT_1^\infty(U)$, for all $U\in\Man$.
Because $\und{G}(U)$ is a discrete category, i.e.\ it only has identity morphisms,
${\AAA_{2}}_U$ is simply a family of $\AQFT_1^\infty(U)$-isomorphisms
${{\AAA_{2}}_U}_g : \AAA_U(\ast)\to \AAA_U(\ast)$ labeled by elements $g\in\und{G}(U)=C^\infty(U,G)$.
The compatibility conditions then state that this labeling is compatible with the 
point-wise group structure on $\und{G}(U)=C^\infty(U,G)$, i.e.\ ${{\AAA_{2}}_U}_{g\cdot g^\prime} = 
{{\AAA_{2}}_U}_{g}\circ{{\AAA_{2}}_U}_{g^\prime}$, for all $g,g^\prime\in \und{G}(U)$,
and ${{\AAA_{2}}_U}_{e} = \id$, for the identity element $e\in\und{G}(U)$.}
From this we obtain a bicategorical cone
\begin{flalign}
\xymatrix{
\ar[d] \und{G} \ar[r] ~&~\und{\{\ast\}}\ar[d]^-{\AAA}\ar@{=>}[dl]_-{\AAA_2}\\
\und{\{\ast\}} \ar[r]_-{\AAA}~&~ \AQFT_1^\infty
}
\end{flalign}
and hence, by the universal property of the loop stack \eqref{eqn:Aut}, 
a stack morphism $\und{G}\to \Aut(\AAA)$ to the smooth automorphism group.
Due to the compatibility conditions of $\AAA_2$ this is a morphism 
of group objects. 
\sk

We conclude this subsection by providing an equivalent, but more explicit,
description of $G$-equivariant smooth AQFTs.
The quotient stack \eqref{eqn:quotientstack} can also be described
as the stackyfication of the prestack $[\{\ast\}/G]_{\mathrm{pre}}: \Man^\op\to \Cat$ 
that assigns to each $U\in\Man$ the groupoid 
\begin{flalign}\label{eqn:quotientprestack}
[\{\ast\}/G]_{\mathrm{pre}}(U)\,=\,\begin{cases}
\mathrm{Obj}: & \ast\\
\mathrm{Mor}: & C^\infty(U,G)
\end{cases}
\end{flalign}
with a single object $\ast$ and morphisms the smooth functions to the Lie group.
(Composition of morphisms is given by the point-wise group structure of $C^\infty(U,G)$.)
On $\Man$-morphisms $h: U\to U^\prime$ this prestack acts via pullback of functions 
$[\{\ast\}/G]_{\mathrm{pre}}(h) := h^\ast$. Because $\astAlg^\infty$ 
is a stack by Proposition \ref{propo:Algstack}, the universal property
of stackyfication implies that the datum of a stack morphism 
$\AAA^{\mathrm{eq}} : [\{\ast\}/G]\to\AQFT_1^\infty$ is equivalent to
a pseudo-natural transformation $ [\{\ast\}/G]_{\mathrm{pre}}\to\AQFT_1^\infty$
between prestacks, or equivalently a pseudo-natural transformation
\begin{flalign}\label{eqn:equivariantAQFTexplicit}
\tilde{\AAA} \,:\, \Loc_1^\infty\times  [\{\ast\}/G]_{\mathrm{pre}}~\longrightarrow~\astAlg^\infty\quad.
\end{flalign}
We will show in Subsection \ref{subsec:fermion} that the latter 
perspective on $G$-equivariant smooth AQFTs is
not very complicated to describe in concrete examples.

%%%%%%%%%%%%%%%%%%%%%%%%%%%%%%%%%%%%%%%%%%%%%%%%

\section{\label{sec:CCR-CAR}Smooth canonical quantization}
The construction of free field theories
in ordinary AQFT crucially relies on
the existence of canonical (anti-)commutation relation 
quantization functors, see e.g.\ \cite{BGP,BG}.
The goal of this section is to show that 
these quantization functors admit a smooth refinement,
which will allow us to construct
both Bosonic and Fermionic examples of
smooth $1$-dimensional AQFTs in Section \ref{sec:examples}.

\subsection{\label{subsec:CCRnew}Canonical commutation relations}
Ordinary canonical commutation relation (CCR) quantization 
is described by a functor $\CCR : \Pois \to\astAlg$ from the 
category of Poisson vector spaces to the category of associative 
and unital $\ast$-algebras. Recall that an object in the category $\Pois$ 
is a tuple $(W,\tau)$, where $W\in\Vec_\bbR$ is a real vector space and $\tau : W\otimes_\bbR W\to \bbR$
is an antisymmetric morphism in $\Vec_\bbR$, and that a morphism $\psi : (W,\tau)\to(W^\prime,\tau^\prime)$ 
in $\Pois$ is a $\Vec_\bbR$-morphism $\psi : W\to W^\prime$ satisfying 
$\tau^\prime\, \circ (\psi\otimes_\bbR \psi) =\tau$. (The
objects in $\Pois$ are interpreted physically as
vector spaces of linear observables, endowed with a Poisson structure.)
The CCR functor assigns to a Poisson vector space $(W,\tau)\in \Pois$
the associative and unital $\ast$-algebra
\begin{flalign}\label{eqn:CCRsimple}
\CCR\big(W,\tau\big)\,:=\, \bigoplus_{n\geq 0} \big(W\otimes_\bbR \bbC\big)^{\otimes_\bbC n}\,\Big/ \,\mathcal{I}_{(W,\tau)}^{\mathrm{CCR}}\,\in\astAlg\quad,
\end{flalign}
where $\mathcal{I}_{(W,\tau)}^{\mathrm{CCR}}$ 
is the $2$-sided $\ast$-ideal generated by the canonical commutation relations
$w\otimes w^\prime - w^\prime \otimes w = \ii\,\tau(w,w^\prime)$, for all $w,w^\prime\in W$,
where $\ii\in\bbC$ denotes the imaginary unit.
The $\ast$-involution on $\CCR(W,\tau)$ is specified by $w^\ast = w$, for all $w\in W$.
Let us reformulate \eqref{eqn:CCRsimple} in a slightly more abstract language. 
For this it is useful to observe that the construction of CCR algebras
\eqref{eqn:CCRsimple} consists of three steps: 
\begin{enumerate}
\item Complexify the real vector space $W\in\Vec_\bbR$ to
the complex vector space $W\otimes_\bbR\bbC\in\Vec_\bbC$, which may be endowed with a
$\ast$-involution $\id\otimes_\bbR \ast : W\otimes_\bbR\bbC \to \overline{W\otimes_\bbR\bbC} = W\otimes_\bbR\overline{\bbC}$
determined by complex conjugation on $\bbC$. Hence, $(W\otimes_\bbR\bbC,\id\otimes_\bbR \ast)\in\astObj(\Vec_\bbC)$
defines a $\ast$-object in the involutive symmetric monoidal category of complex vector spaces.

\item Take the free order-reversing $\ast$-monoid
of $(W\otimes_\bbR\bbC,\id \otimes_\bbR\ast)\in\astObj(\Vec_\bbC)$, which defines
the associative and unital $\ast$-algebra 
$\bigoplus_{n\geq 0} \big(W\otimes_\bbR \bbC\big)^{\otimes_\bbC n}\in \astAlg$.

\item Implement the canonical commutation relations associated with the Poisson structure
$\tau$ by a coequalizer in the category $\astAlg$.
\end{enumerate}

Before we can generalize this construction to the context of stacks, we
have to find a smooth refinement of the category $\Pois$.
As explained in Subsection \ref{subsec:Algstack}, we consider the stack 
$\Sh_{C^\infty_\bbR}$ of sheaves of $C^\infty_\bbR$-modules as a smooth refinement
of the category $\Vec_\bbR$, hence a smooth refinement of the category $\Pois$
should be built from this stack. Concretely, we define the (pre)stack
$\Pois^\infty :\Man^\op \to\Cat$ by the following data.
To each manifold $U\in \Man$, it assigns the category
$\Pois^\infty(U)$ whose objects are tuples
$(W,\tau)$ with $W\in \Sh_{C^\infty_\bbR}(U)$ and 
$\tau : W\otimes_{C^\infty_{\bbR,U}} W\to C^\infty_{\bbR,U}$ an antisymmetric
morphism in $\Sh_{C^\infty_\bbR}(U)$, called Poisson structure. The morphisms $\psi : (W,\tau)\to (W^\prime,\tau^\prime)$ in this category
are $\Sh_{C^\infty_\bbR}(U)$-morphisms $\psi : W\to W^\prime$ satisfying
$\tau^\prime \circ (\psi\otimes_{C^\infty_{\bbR,U}} \psi) =\tau$.
To a morphism $h: U\to U^\prime$ in $\Man$, the prestack $\Pois^\infty$ assigns
the functor
\begin{flalign}
h^\ast \,:=\, \Pois^\infty(h)~:~\Pois^\infty(U^\prime)~\longrightarrow~\Pois^\infty(U)
\end{flalign}
that assigns to $(W,\tau)\in \Pois^\infty(U^\prime)$ the object
in $\Pois^\infty(U)$ determined by the object $h^\ast W\in \Sh_{C^\infty_\bbR}(U)$
(see \eqref{eqn:sheafmodstack2}) and the Poisson structure
\begin{flalign}
\xymatrix{
(h^\ast W)\otimes_{C^\infty_{\bbR,U}}  ( h^\ast W) \,\cong\, 
h^\ast\big(W\otimes_{C^\infty_{\bbR,U^\prime}} W\big) \ar[r]^-{h^\ast \tau}~&~
h^\ast C^\infty_{\bbR,U^\prime} \,\cong\,C^\infty_{\bbR,U}
}\quad,
\end{flalign}
where $\cong$ are the coherence isomorphisms in \eqref{eqn:ShSMcoherences}.
\begin{propo}\label{prop:PoVecStack}
The prestack $\Pois^\infty : \Man^\op\to \Cat$ 
defined above is a stack, i.e.\ it satisfies
the descent condition from Definition \ref{def:stack}.
\end{propo}
\begin{proof}
This follows from the fact that $\Sh_{C^\infty_\bbR}$ is a stack, see 
Proposition \ref{prop:Shmodulestack}. Indeed, spelling out descent 
for objects $(W,\tau) \in \Pois^\infty(U)$, one observes that it involves descent for the underlying objects 
$W\in \Sh_{C^\infty_\bbR}(U)$ 
and also for the underlying $\Sh_{C^\infty_\bbR}(U)$-morphisms 
$\tau : W\otimes_{C^\infty_{\bbR,U}}W\to C^\infty_{\bbR,U}$,
which are both simple consequences of descent for the stack $\Sh_{C^\infty_\bbR}$. 
Similarly, descent for  $\Pois^\infty(U)$-morphisms 
$\psi : (W,\tau)\to (W^\prime,\tau^\prime)$ 
involves descent for the underlying 
$\Sh_{C^\infty_\bbR}(U)$-morphisms $\psi: W \to W^\prime$ 
and the verification that $\tau^\prime \circ (\psi\otimes_{C^\infty_{\bbR,U}} \psi) =\tau$ 
coincide as $\Sh_{C^\infty_\bbR}(U)$-morphisms, 
which are again both consequences of descent for the stack 
$\Sh_{C^\infty_\bbR}$ and of the fact that 
the descent data have this property. 
\end{proof}

Adopting an analogous three step construction as in the case of the 
ordinary CCR functor, we shall now define a stack morphism
\begin{flalign}
\CCR\,:\, \Pois^\infty ~\longrightarrow~\astAlg^\infty
\end{flalign}
that provides a smooth refinement of CCR quantization.
By Definition \ref{def:stackmorphism}, this consists of functors
\begin{flalign}\label{eqn:CCRfunctorU}
\CCR_U \,:\, \Pois^\infty(U)~\longrightarrow~\astAlg^\infty(U)\quad,
\end{flalign}
for each manifold $U\in\Man$, together with coherence isomorphisms.
Regarding the first step, we observe that, for each $U\in\Man$, there exists an adjunction
\begin{flalign}\label{eqn:LRadjunction}
\xymatrix{
L_U \,:\, \Sh_{C^\infty_\bbR}(U) ~\ar@<0.5ex>[r] & \ar@<0.5ex>[l]~ \astObj\big(\Sh_{C^\infty_\bbC}(U)\big)\,:\, R_U
}\quad.
\end{flalign}
The left adjoint functor $L_U$ assigns to $W\in \Sh_{C^\infty_\bbR}(U)$ its complexification
$W\otimes_{C^\infty_{\bbR,U}} C^\infty_{\bbC,U} \in  \Sh_{C^\infty_\bbC}(U)$, with the $\ast$-object structure
$\id\otimes_{C^\infty_{\bbR,U}} \ast$ determined by complex conjugation 
$\ast: C^\infty_{\bbC,U}\to \overline{C^\infty_{\bbC,U}}$. 
The right adjoint functor $R_U$ assigns to a $\ast$-object $(V,\ast)$ in $\Sh_{C^\infty_\bbC}(U)$
the sheaf of $\ast$-invariants $R_U(V,\ast) =
\mathrm{ker}\big(V_\bbR \stackrel{\ast-\id}{\longrightarrow} V_\bbR\big)
\in \Sh_{C^\infty_\bbR}(U)$, where
by $V_\bbR\in \Sh_{C^\infty_\bbR}(U)$ we denote the restriction of $V\in \Sh_{C^\infty_\bbC}(U)$ 
to a sheaf of $C^\infty_{\bbR,U}$-modules via the morphism
$C^\infty_{\bbR,U}\to C^\infty_{\bbC,U}$ from real to complex-valued functions.
\sk

Regarding the second step, we observe that, for each $U\in\Man$, there exists an adjunction
\begin{flalign}\label{eqn:FGadjunction}
\xymatrix{
F_U \,:\, \astObj\big(\Sh_{C^\infty_\bbC}(U)\big) ~\ar@<0.5ex>[r] & \ar@<0.5ex>[l]~ \astAlg^\infty(U)\,:\, G_U
}\quad.
\end{flalign}
The right adjoint functor $G_U$ assigns to an associative and unital $\ast$-algebra $(A,\mu,\eta,\ast)$ 
in $\Sh_{C^\infty_\bbC}(U)$ its underlying $\ast$-object $(A,\ast)$, i.e.\ it forgets the multiplication $\mu$
and unit $\eta$. The left adjoint functor $F_U$ is the free order-reversing $\ast$-monoid 
functor. Explicitly, it assigns to a $\ast$-object $(V,\ast)$ the free order-reversing $\ast$-monoid 
$F_U(V,\ast) := \bigoplus_{n\geq 0} V^{\otimes n}$,
where tensor products and coproducts are formed in the symmetric monoidal category 
$\Sh_{C^\infty_\bbC}(U)$. The order-reversing $\ast$-structure of $F_U(V,\ast)$ is
defined by the canonical extension of the $\ast$-structure on the generators $(V,\ast)$.
\sk

With these preparations, we can now define the values of  \eqref{eqn:CCRfunctorU} 
on objects by carrying out the third step.
Explicitly, given any object $(W,\tau)\in \Pois^\infty(U)$, i.e.\ $W\in \Sh_{C^\infty_\bbR}(U)$
and $\tau: W\otimes_{C^\infty_{\bbR,U}}W \to C^\infty_{\bbR,U}$ an antisymmetric morphism
in $\Sh_{C^\infty_\bbR}(U)$, we define
\begin{flalign}\label{eqn:CCRU}
\CCR_U\big(W,\tau\big):= \colim\Big(
\xymatrix{
F_UL_U\big(W\otimes_{C^\infty_{\bbR,U}} W\big) \ar@<0.5ex>[r]^-{r_1}\ar@<-0.5ex>[r]_-{r_2} ~&~F_UL_U\big(W\big)
}\Big)\,\in\,\astAlg^\infty(U)
\end{flalign}
by a coequalizer in $\astAlg^\infty(U)$. The relations $r_1,r_2$ are defined in terms of 
their adjuncts under the adjunctions in \eqref{eqn:LRadjunction} and \eqref{eqn:FGadjunction} by
\begin{subequations}
\begin{flalign}
\xymatrix{
\ar[d]_-{ \text{unit }L\dashv R}W\otimes_{C^\infty_{\bbR,U}} W \ar[rr]^-{\tilde{r}_1} ~&~ ~&~ RGFL(W)\\
RL\big(W\otimes_{C^\infty_{\bbR,U}} W\big)\ar[r]_-{\cong}~&~R\big(L(W)\otimes_{C^\infty_{\bbC,U}}L(W)\big) \ar[r]_-{\text{unit }F\dashv G}~&~
R\big(GFL(W)\otimes_{C^\infty_{\bbC,U}}GFL(W)\big)\ar[u]_-{R(\mu-\mu^\op)}
}
\end{flalign}
and
\begin{flalign}
\xymatrix{
\ar[d]_-{\tau}W\otimes_{C^\infty_{\bbR,U}} W \ar[rr]^-{\tilde{r}_2} ~&~~&~ RGFL(W)\\
C^\infty_{\bbR,U} \ar[r]_-{\text{unit }L\dashv R}~&~RL\big(C^\infty_{\bbR,U}\big) \ar[r]_-{\cong}~&~R\big(C^\infty_{\bbC,U},\ast\big)\ar[u]_-{R(\ii\,\eta)}
}
\end{flalign}
\end{subequations}
where we suppressed for notational convenience the subscripts ${}_U$ on the functors.
Here $\mu^{(\op)}$ denotes the (opposite) multiplication and $\eta$ the unit element in $FL(W)$.
Because the coequalizer in \eqref{eqn:CCRU} is clearly functorial
with respect to morphisms $\psi : (W,\tau)\to (W^\prime,\tau^\prime)$ in $\Pois^\infty(U)$,
we have successfully defined the desired functor in \eqref{eqn:CCRfunctorU}.
\begin{rem}
For $U=\{\ast\}$ a point, \eqref{eqn:CCRU} gives precisely the usual CCR algebra \eqref{eqn:CCRsimple}.
\end{rem}

To complete our construction of the desired stack morphism $\CCR : \Pois^\infty\to \astAlg^\infty$,
it remains to define coherence isomorphisms (see Definition \ref{def:stackmorphism})
\begin{flalign}\label{eqn:CCRcoherence}
\xymatrix@C=4em{
\ar[d]_-{h^\ast} \Pois^\infty(U^\prime) \ar[r]^-{\CCR_{U^\prime}}~&~\astAlg^\infty(U^\prime)\ar[d]^-{h^\ast} \ar@{=>}[dl]_-{\CCR_h~~}\\
\Pois^\infty(U) \ar[r]_-{\CCR_U}~&~\astAlg^\infty(U)
}
\end{flalign}
for all morphisms $h:U\to U^\prime$ in $\Man$. These can be built from the analogous
coherence isomorphisms for the left adjoint functors in \eqref{eqn:LRadjunction} 
and \eqref{eqn:FGadjunction}, i.e.\
\begin{flalign}\label{eqn:LFcoherences}
\xymatrix@C=4em{
\ar[d]_-{h^\ast}\Sh_{C^\infty_\bbR}(U^\prime) \ar[r]^-{L_{U^\prime}}~&~\ar@{=>}[dl]_-{L_h}\ar[d]_-{h^\ast}\astObj\big(\Sh_{C^\infty_{\bbC}}(U^\prime)\big) 
\ar[r]^-{F_{U^\prime}}~&~ \ar@{=>}[dl]_-{F_h}\astAlg^\infty(U^\prime)\ar[d]^-{h^\ast}\\
\Sh_{C^\infty_\bbR}(U) \ar[r]_-{L_{U}}~&~\astObj\big(\Sh_{C^\infty_{\bbC}}(U)\big)\ar[r]_-{F_{U}} ~&~ \astAlg^\infty(U)\\
}
\end{flalign}
Explicitly, for $W\in \Sh_{C^\infty_\bbR}(U^\prime)$, the isomorphism $L_h$ is given by
\begin{subequations}
\begin{flalign}
\nn h^\ast L_{U^{\prime}} (W) &=h^\ast\big(W\otimes_{C^\infty_{\bbR,U^\prime}} C^\infty_{\bbC,U^\prime},\id\otimes \ast\big)\\
\nn &\cong
\Big(h^{-1}(W)\otimes_{h^{-1}(C^\infty_{\bbR,U^\prime})} h^{-1}(C^\infty_{\bbC,U^\prime})\otimes_{h^{-1}(C^\infty_{\bbC,U^\prime})} C^\infty_{\bbC,U},\id\otimes\ast\otimes\ast  \Big)\\
\nn &\cong\Big(h^{-1}(W)\otimes_{h^{-1}(C^\infty_{\bbR,U^\prime})} C^\infty_{\bbC,U},\id\otimes\ast  \Big)\\
&\cong \Big(h^{-1}(W)\otimes_{h^{-1}(C^\infty_{\bbR,U^\prime})} C^\infty_{\bbR,U} \otimes_{ C^\infty_{\bbR,U}} 
C^\infty_{\bbC,U},\id\otimes\id \otimes\ast  \Big)= L_Uh^\ast(W)\quad.
\end{flalign}
For $(V,\ast)\in \astObj\big(\Sh_{C^\infty_{\bbC}}(U^\prime)\big) $, the isomorphism
$F_h$ is given by
\begin{flalign}
\nn h^\ast F_{U^\prime}(V,\ast)& = h^\ast\Big(\bigoplus_{n\geq 0} V^{\otimes_{C^\infty_{\bbC,U^\prime}} n} \Big)
\cong \bigoplus_{n\geq 0} h^\ast\Big(V^{\otimes_{C^\infty_{\bbC,U^\prime}} n}\Big)\\
& \cong \bigoplus_{n\geq 0} \big(h^\ast(V)\big)^{\otimes_{C^\infty_{\bbC,U}} n} = F_Uh^\ast(V,\ast)\quad,
\end{flalign}
\end{subequations}
where in the second step we have used that $h^\ast$ preserves coproducts because it is a left adjoint functor 
and in the third step we have used the coherence isomorphisms of the involutive symmetric monoidal 
stack $\Sh_{C^\infty_\bbC}$ from Corollary \ref{cor:ShstackISMCAT}. Pasting the natural isomorphisms
in \eqref{eqn:LFcoherences} defines a natural isomorphism
$(FL)_h : h^\ast F_{U^\prime}L_{U^\prime} \Rightarrow F_U L_U h^\ast$. For
every object $(W,\tau)\in \Pois^\infty(U^\prime)$, the associated isomorphism
$h^\ast F_{U^\prime} L_{U^\prime}(W)\cong F_U L_U h^\ast(W)$ 
descends to the CCR algebras \eqref{eqn:CCRU} and thereby 
defines  the natural isomorphism $\CCR_h$ in \eqref{eqn:CCRcoherence}.
\begin{propo}\label{prop:CCR}
The construction above defines a stack morphism
$\CCR : \Pois^\infty\to\astAlg^\infty$.
\end{propo}

\subsection{\label{subsec:CARnew}Canonical anti-commutation relations}
A smooth refinement of the canonical anti-commutation relation 
(CAR) quantization functor for Fermionic theories
can be developed along the same lines as in Subsection \ref{subsec:CCRnew}.
Before we spell out some of the details, let
us briefly recall the ordinary CAR functor
$\CAR : \IPVec \to\astAlg$ following the 
presentation in \cite{Dappiaggi}.
The category $\IPVec$ has objects $(V,\ast,\langle\cdot,\cdot\rangle)$
consisting of a $\ast$-object $(V,\ast)\in \astObj(\Vec_\bbC)$
in the involutive symmetric monoidal category $\Vec_\bbC$
and a symmetric $\ast$-morphism $\langle\cdot, \cdot\rangle : (V,\ast)\otimes (V,\ast)\to (\bbC,\ast)$.
More explicitly, the latter is a symmetric $\bbC$-linear map $\langle \cdot,\cdot\rangle : V\otimes V\to \bbC$
satisfying
\begin{flalign}\label{eqn:astFermionCompatible}
\xymatrix@C=4em{
\ar[d]_-{\ast\otimes\ast}V\otimes V \ar[r]^-{\langle\cdot,\cdot\rangle}~&~\bbC\ar[d]^-{\ast}\\
\overline{V}\otimes\overline{V}\cong \overline{V\otimes V} \ar[r]_-{\overline{\langle\cdot,\cdot\rangle}}~&~\overline{\bbC}
}
\end{flalign}
or at the level of elements 
$\langle v,v^\prime\rangle^\ast = \langle v^\ast,v^{\prime\ast}\rangle$, for all $v,v^\prime\in V$.
Morphisms $\psi : (V,\ast,\langle\cdot,\cdot\rangle) \to (V^\prime,\ast^\prime,\langle\cdot,\cdot\rangle^\prime)$
in  $\IPVec$ are $\ast$-morphisms $\psi : (V,\ast)\to (V^\prime,\ast^\prime)$ satisfying
$\langle\cdot,\cdot\rangle^\prime\circ  (\psi\otimes \psi) = \langle\cdot,\cdot\rangle$.
The CAR functor assigns to $(V,\ast,\langle\cdot,\cdot\rangle)\in\IPVec$
the associative and unital $\ast$-algebra
\begin{flalign}
\CAR\big(V,\ast,\langle\cdot,\cdot\rangle\big)\,:=\, \bigoplus_{n\geq 0} V^{\otimes n}\,\Big/ \,\mathcal{I}_{(V,\ast ,\langle\cdot,\cdot\rangle)}^{\mathrm{CAR}}\,\in\astAlg\quad,
\end{flalign}
where $\mathcal{I}_{(V,\ast,\langle\cdot,\cdot\rangle)}^{\mathrm{CAR}}$ is
the $2$-sided $\ast$-ideal generated by the canonical anti-commutation relations
$v\otimes v^\prime + v^\prime\otimes v = \langle v,v^\prime\rangle$, for all $v,v^\prime\in V$.
Observe that this construction consists of two steps:
\begin{enumerate}
\item Take the free order-reversing $\ast$-monoid
of $(V,\ast)\in\astObj(\Vec_\bbC)$, which defines
the associative and unital $\ast$-algebra 
$\bigoplus_{n\geq 0} V^{\otimes n}\in \astAlg$.

\item Implement the canonical anti-commutation relations associated with 
$\langle\cdot,\cdot\rangle$ by a coequalizer in the category $\astAlg$.
\end{enumerate}

To obtain a smooth refinement of the category $\IPVec$, 
we follow the same strategy as in Subsection \ref{subsec:CCRnew}.
We define a (pre)stack $\IPVec^\infty : \Man^\op\to \Cat$ by the following data.
To each manifold $U\in\Man$, it assigns the category $\IPVec^\infty(U)$
whose objects $(V,\ast,\langle\cdot,\cdot\rangle)$ consist
of a $\ast$-object $(V,\ast)\in \astObj(\Sh_{C^\infty_\bbC}(U))$
and a symmetric $\ast$-morphism $\langle\cdot, \cdot\rangle : (V,\ast)\otimes_{C^\infty_{\bbC,U}} 
(V,\ast)\to (C^\infty_{\bbC,U},\ast)$.
The morphisms $\psi : (V,\ast,\langle\cdot,\cdot\rangle)\to 
(V^\prime,\ast^\prime,\langle\cdot,\cdot\rangle^\prime) $ in this category
are $\astObj(\Sh_{C^\infty_\bbC}(U))$-morphisms $\psi : (V,\ast)\to 
(V^\prime,\ast^\prime)$ satisfying $\langle\cdot,\cdot\rangle^\prime\circ (\psi\otimes_{C^\infty_{\bbC,U}} \psi)
=\langle\cdot,\cdot\rangle$. 
To a morphism $h: U\to U^\prime$ in $\Man$, the prestack $\IPVec^\infty$ assigns
the functor
\begin{flalign}
h^\ast \,:=\, \IPVec^\infty(h)~:~\IPVec^\infty(U^\prime)~\longrightarrow~\IPVec^\infty(U)
\end{flalign}
that assigns to $(V,\ast,\langle\cdot,\cdot\rangle)\in \IPVec^\infty(U^\prime)$ the object
in $\IPVec^\infty(U)$ determined by the object $h^\ast (V,\ast)\in \astObj(\Sh_{C^\infty_\bbC}(U))$
and the morphism
\begin{flalign}
\xymatrix{
(h^\ast (V,\ast))\otimes_{C^\infty_{\bbC,U}}  (h^\ast (V,\ast) )\,\cong\, 
h^\ast\big((V,\ast)\otimes_{C^\infty_{\bbC,U^\prime}} (V,\ast)\big) \ar[r]^-{h^\ast \langle\cdot,\cdot\rangle}~&~
h^\ast (C^\infty_{\bbC,U^\prime},\ast) \,\cong\,(C^\infty_{\bbC,U},\ast)
}\quad,
\end{flalign}
where we have used the coherence isomorphisms of the involutive symmetric monoidal 
stack $\Sh_{C^\infty_\bbC}$ from Corollary \ref{cor:ShstackISMCAT}.
The proof of the following statement is completely analogous to the one of 
Proposition \ref{prop:PoVecStack}.
\begin{propo}\label{prop:IPVecStack}
The prestack $\IPVec^\infty : \Man^\op\to \Cat$ 
defined above is a stack, i.e.\ it satisfies
the descent condition from Definition \ref{def:stack}.
\end{propo}

We shall now define a stack morphism
\begin{flalign}
\CAR\,:\, \IPVec^\infty ~\longrightarrow~\astAlg^\infty
\end{flalign}
that provides a smooth refinement of CAR quantization.
In analogy to \eqref{eqn:CCRU}, we define the component functors
$\CAR_U : \IPVec^\infty(U)\to \astAlg^\infty(U)$, for all $U\in \Man$, by the coequalizer
\begin{flalign}
\CAR_U\big(V,\ast,\langle\cdot,\cdot\rangle\big):= \colim\Big(
\xymatrix{
F_U\big((V,\ast)\otimes_{C^\infty_{\bbC,U}} (V,\ast)\big) \ar@<0.5ex>[r]^-{s_1}\ar@<-0.5ex>[r]_-{s_2} ~&~F_U\big(V,\ast\big)
}\Big)\,\in\,\astAlg^\infty(U)\quad,
\end{flalign}
where the relations $s_1,s_2$ are defined in terms of their adjuncts under 
\eqref{eqn:FGadjunction} by
\begin{subequations}
\begin{flalign}
\xymatrix{
\ar[d]_-{\text{unit }F\dashv G}(V,\ast)\otimes_{C^\infty_{\bbC,U}} (V,\ast)\ar[r]^-{\tilde{s}_1} ~&~ GF(V,\ast)\\
GF(V,\ast)\otimes_{C^\infty_{\bbC,U}} GF(V,\ast) \ar[ru]_-{\mu + \mu^\op} ~&~
}
\end{flalign}
and
\begin{flalign}
\xymatrix{
\ar[d]_-{\langle\cdot,\cdot\rangle}(V,\ast)\otimes_{C^\infty_{\bbC,U}} (V,\ast) \ar[r]^-{\tilde{s}_2}~&~ GF(V,\ast)\\
(C^\infty_{\bbC,U},\ast) \ar[ru]_-{\eta}~&~
}
\end{flalign}
\end{subequations}
where we suppressed for notational convenience the subscripts ${}_U$ on the functors.
The coherence isomorphisms
\begin{flalign}
\xymatrix@C=4em{
\ar[d]_-{h^\ast} \IPVec^\infty(U^\prime) \ar[r]^-{\CAR_{U^\prime}}~&~\astAlg^\infty(U^\prime)\ar[d]^-{h^\ast} \ar@{=>}[dl]_-{\CAR_h~~}\\
\IPVec^\infty(U) \ar[r]_-{\CAR_U}~&~\astAlg^\infty(U)
}
\end{flalign}
associated with $\Man$-morphisms $h: U\to U^\prime$ are built similarly 
to those in Subsection \ref{subsec:CCRnew}. Summing up, we have
\begin{propo}\label{prop:CAR}
The construction above defines a stack morphism
$\CAR : \IPVec^\infty\to\astAlg^\infty$.
\end{propo}

%%%%%%%%%%%%%%%%%%%%%%%%%%%%%%%%%%%%%%%%%%%%%%%%

\section{\label{sec:examples}Illustration through free theories}
We shall illustrate our formalism by constructing 
concrete examples of smooth $1$-dimensional AQFTs.
The models we study are smooth refinements
of the Bosonic and Fermionic free field theories discussed 
in e.g.\ \cite{BG,BGP}. Similarly to the ordinary case,
our Bosonic models will be described
by stack morphisms 
\begin{flalign}\label{eqn:AAAviaCCR}
\xymatrix{
\ar[dr]_-{\LLL^{\mathsf{b}}}\Loc_1^\infty \ar[rr]^-{\AAA^{\mathsf{b}}}~&~~&~\astAlg^\infty\\
~&~ \Pois^\infty \ar[ur]_-{\CCR}~&~
}
\end{flalign}
obtained as the composition of a stack morphism $\LLL^\mathsf{b}$ 
assigning the linear observables with their Poisson structure
and the $\CCR$-quantization stack morphism developed in Subsection \ref{subsec:CCRnew}.
The Fermionic models will be described similarly by stack morphisms
\begin{flalign}\label{eqn:AAAviaCAR}
\xymatrix{
\ar[dr]_-{\LLL^{\mathsf{f}}}\Loc_1^\infty \ar[rr]^-{\AAA^{\mathsf{f}}}~&~~&~\astAlg^\infty\\
~&~ \IPVec^\infty \ar[ur]_-{\CAR}~&~
}
\end{flalign}
factorizing through the $\CAR$-quantization stack morphism developed 
in Subsection \ref{subsec:CARnew}.
\sk

Inspired by the standard constructions in ordinary AQFT \cite{BG,BGP},
we shall obtain examples of the stack morphisms $\LLL^{\mathsf{b}/\mathsf{f}}$ 
assigning linear observables by using a suitable smooth refinement
of the concept of retarded/advanced Green operators $G^\pm$ to be
developed in Subsection \ref{subsec:Green} below.
Recall that the role of such Green operators is to determine
the Poisson structure $\tau$ of a Bosonic theory 
and the bilinear map $\langle\cdot,\cdot\rangle$ 
of a Fermionic theory. In Subsection \ref{subsec:example}
we will spell out this construction for the simplest case
of a $1$-dimensional massive scalar field, which is
equivalent to the harmonic oscillator\footnote{In the context of quantum 
mechanics, the harmonic oscillator is usually described via the equal time 
canonical commutation relations $[\hat{x},\hat{p}]=\ii \,\hat{\oone}$ and the Hamiltonian
$\hat{H} = \frac{1}{2}\hat{p}^2 + \frac{m^2}{2}\hat{x}^2$ with frequency/mass parameter $m>0$.
In the context of AQFT, one works instead with the covariant commutation relations 
$[\hat{\Phi}(\varphi),\hat{\Phi}(\varphi^\prime)] = \ii\,\int_\bbR \varphi \,G(\varphi^\prime)\,
\dd t\,\hat{\oone}$, where $\hat{\Phi}(\varphi^{(\prime)})$ are the field operators smeared by
compactly supported functions $\varphi,\varphi^\prime\in C^\infty_\cc(\bbR)$ on the time line $\bbR$
and $G = G^+ - G^-$ is the retarded-minus-advanced Green operator for the equation of motion operator
$P = \partial_t^2 + m^2$. Due to the well-posed initial value problem, one can show that both 
approaches are equivalent, see e.g.\ \cite[Remark 3.3.4]{BD}.
}. We shall even construct a smooth $\tilde{U}$-family of smooth AQFTs 
(in the sense of \eqref{eqn:tildeUfamilyAQFT}) that describes a family
of $1$-dimensional massive scalar fields with a smoothly varying
mass parameter $m\in C^\infty(\tilde{U},\bbR^{>0})$.
In Subsection \ref{subsec:fermion} we construct the $1$-dimensional massless Dirac
field as a smooth AQFT and show that its global $U(1)$-symmetry
is realized in terms of smooth automorphisms in the sense of \eqref{eqn:Aut}.

\subsection{\label{subsec:Green}Green operators, solutions and initial data}
Let us consider a manifold $U \in \Man$ and a smooth $U$-family of 
$1$-dimensional spacetimes $(\pi:M\to U,E) \in \Loc_1^\infty(U)$. 
We introduce the functor 
\begin{flalign}\label{eq:Cinfti}
C^\infty_{\pi}\,:\, \Open(U)^\op~\longrightarrow~\Set
\end{flalign}
that assigns to each open subset $U^\prime \subseteq U$ 
the set $C^\infty_{\pi}(U^\prime) := C^\infty_\bbR(M\vert_{U^\prime})$ of real valued smooth functions 
on the restricted total space $M\vert_{U^\prime} = \pi^{-1}(U^\prime)\subseteq M$
and to each open subset inclusion $U^\prime\subseteq U^{\prime\prime}\subseteq U$
the restriction map $C^\infty_\bbR(M\vert_{U^{\prime\prime}})\to C^\infty_\bbR(M\vert_{U^\prime})$.
Together with the $C^\infty_{\bbR,U}$-module structure
induced by pullback of functions along the projection map $\pi:M\to U$,
this defines an object $C^\infty_\pi \in \Sh_{C^\infty_\bbR}(U)$
that we shall interpret as the field configuration space of a real scalar field
on the smooth family of spacetimes $(\pi:M\to U,E) \in \Loc_1^\infty(U)$.
More generally, the configuration space of vector-valued fields
is given by
\begin{flalign}
C^\infty_{\pi}\otimes\bbK^n\,\in \,  \Sh_{C^\infty_\bbK}(U)\quad,
\end{flalign}
where $n\in\bbZ_{\geq 1}$ is the number of field components,
and we take $\bbK=\bbR$ for real fields and $\bbK=\bbC$ 
for complex fields. As equation of motion we will consider a 
$\Sh_{C^\infty_\bbK}(U)$-morphism 
$P: C^\infty_{\pi}\otimes\bbK^n \to C^\infty_{\pi}\otimes\bbK^n$ 
given by a vertical differential operator on $\pi : M\to U$,
i.e.\ a differential operator on $M$ that differentiates only along the fibers
of $\pi : M\to U$. See our Examples \ref{ex:oscillator} and \ref{ex:Dirac} below.
\sk

In order to define a concept of Green operators for such $P$, 
we introduce certain subsheaves of the sheaf of functions $C^\infty_\pi$ on $\pi:M\to U$
that describe functions with restrictions on their vertical support.
In the following definition, we shall use that the fiber bundle $\pi : M\to U$
underlying any object $(\pi:M\to U,E) \in \Loc_1^\infty(U)$
admits sections because the fibers are open intervals,
see e.g.\ \cite[Sections 12.2 and 6.7]{Steenrod}.
Furthermore, given any subset $S\subseteq M$ 
of the total space, we denote by $J^{\pm}_\ver(S)\subseteq M$ 
the {\em vertical future/past} of $S$, i.e.\ the subset of all points that can 
be reached from $S$ by future/past directed vertical curves
with respect to the orientation induced by $E\in\Omega^1_\ver(M)$.
\begin{defi}\label{def:supportrestrictions}
Let $(\pi:M\to U,E) \in \Loc_1^\infty(U)$ and $U^\prime\subseteq U$ an open subset.
We say that a function $\varphi\in C^\infty_{\pi}(U^\prime) = C^\infty_\bbR(M\vert_{U^\prime})$ 
is {\em vertically past/future compactly supported} if there exists a section
$\sigma : U^\prime \to M\vert_{U^\prime}$ such that 
$\supp(\varphi) \subseteq J^\pm_\ver(\sigma(U^\prime))$.
We say that $\varphi\in C^\infty_{\pi}(U^\prime) $ is {\em vertically compactly supported} 
if it is both vertically past and future compactly supported, i.e.\ there exist
two sections $\sigma_1,\sigma_2 :  U^\prime \to M\vert_{U^\prime}$
such that $\supp (\varphi) \subseteq J_\ver^+(\sigma_1(U^\prime)) \cap J_\ver^-(\sigma_2(U^\prime))  $.
\end{defi}
\begin{rem}
Note that our definition of vertically compactly supported functions
uses manifestly the fact that we consider smooth families of 
{\em $1$-dimensional} spacetimes.
In this $1$-dimensional case, we have bundles $\pi:M\to U$ whose fibers 
are intervals, hence it makes sense to define vertical compactness through 
vertical boundedness from above and below. A dimension-independent definition for
$\varphi\in C^\infty_{\pi}(U^\prime) $ to be vertically compactly supported
is given by the condition that $\supp(\varphi)\cap \pi^{-1}(K)$ is compact, for all
$K\subseteq U^\prime$ compact. Upon sheafification (see Definition \ref{def:supportsheaves} below),
this coincides in the $1$-dimensional case with our more practical Definition \ref{def:supportrestrictions}.
\end{rem}

From this definition we obtain sub-{\em pre}sheaves 
$\widetilde{C}^\infty_{\pi\, \vc}$, $\widetilde{C}^\infty_{\pi\, \vpc}$
and $\widetilde{C}^\infty_{\pi\, \vfc}$ of $C^\infty_\pi$ that assign
vertically compactly supported, vertically past compactly supported and 
vertically future compactly supported functions.
Note that these presheaves are separated, but they do not satisfy the descent condition
for sheaves and hence have to be sheafified.
\begin{defi}\label{def:supportsheaves}
We denote by $C^\infty_{\pi\,\vc}, C^\infty_{\pi\,\vpc}, C^\infty_{\pi\,\vfc}\in 
\Sh_{C^\infty_\bbR}(U)$ the sheafifications
of the presheaves $\widetilde{C}^\infty_{\pi\, \vc}$ of vertically compactly supported functions, 
$\widetilde{C}^\infty_{\pi\, \vpc}$ of vertically past compactly supported functions and
$\widetilde{C}^\infty_{\pi\, \vfc}$ of vertically future compactly supported functions. 
\end{defi}
\begin{rem}
These sheaves admit the following explicit description as subsheaves of $C^\infty_\pi$.
To each open subset $U^\prime\subseteq U$, the sheaf $C^\infty_{\pi\, \mathrm{v(p/f)c}}$
assigns the subset $C^\infty_{\pi\,\mathrm{v(p/f)c}}(U^\prime) \subseteq C^\infty_\pi(U^\prime)$
consisting of all functions $\varphi\in C^\infty_\pi(U^\prime) $ 
that satisfy the following local support condition:
For every point $x\in U^\prime$, there exists an open neighborhood 
$U_x\subseteq U^\prime$ of $x$ such that the restriction $\varphi\vert_{U_x}\in C^\infty_\pi(U_x)$ is 
vertically (past/future) compactly supported in the sense of Definition \ref{def:supportrestrictions}.
\end{rem}

With these preparations we can now introduce a concept of Green operators.
\begin{defi}\label{def:Green}
Let $P: C^\infty_{\pi} \otimes\bbK^n \to C^\infty_{\pi}\otimes\bbK^n$ 
be a $\Sh_{C^\infty_\bbK}(U)$-morphism
that is determined from a vertical differential operator on $\pi:M\to U$, 
i.e.\ a differential operator on $M$ that differentiates only along the fibers
of $\pi : M\to U$.
A {\em retarded/advanced Green operator} for $P$ 
is a $\Sh_{C^\infty_\bbK}(U)$-morphism
$G^\pm :  C^\infty_{\pi\,\vpc/\vfc}\otimes\bbK^n \to C^\infty_{\pi\,\vpc/\vfc}\otimes\bbK^n$ 
that satisfies the following properties:
\begin{enumerate}
\item[(i)] $G^\pm$ is the inverse of the restriction
$P: C^\infty_{\pi\,\vpc/\vfc}\otimes\bbK^n \to C^\infty_{\pi\,\vpc/\vfc}\otimes\bbK^n$ 
of $P$ to the subsheaves of vertically past/future compactly supported functions.

\item[(ii)] For each open subset $U^\prime \subseteq U$
and $\varphi\in C^\infty_{\pi\,\vpc/\vfc}(U^\prime)\otimes\bbK^n$, 
we have $\supp(G^\pm \varphi) \subseteq J^\pm_\ver(\supp(\varphi))$. 
\end{enumerate}
We refer to the $\Sh_{C^\infty_\bbK}(U)$-morphism
$G := G^+ - G^-: C^\infty_{\pi\,\vc}\otimes\bbK^n \to C^\infty_{\pi}\otimes\bbK^n$ 
as the {\em causal propagator}. 
\end{defi}
\begin{rem}
Observe that, as a consequence of item (i), 
retarded and advanced Green operators are unique, provided they exist. 
Their existence is instead a condition on $P$, namely the restrictions 
$P: C^\infty_{\pi\,\vpc/\vfc}\otimes\bbK^n \to C^\infty_{\pi\,\vpc/\vfc}\otimes\bbK^n$ 
must be invertible and their inverses must fulfill also item (ii). 
Examples \ref{ex:oscillator}, \ref{ex:familyoscillator} and \ref{ex:Dirac}
present vertical differential operators that fulfill these conditions. 
\end{rem}

The usual exact sequence for $P$ and $G$, see e.g.\ \cite{BGP},
generalizes to our context.
\begin{propo}\label{prop:exactsequence}
Let $P: C^\infty_{\pi}\otimes\bbK^n \to C^\infty_{\pi}\otimes\bbK^n$ be a 
$\Sh_{C^\infty_\bbK}(U)$-morphism
that is determined from a vertical differential operator on $\pi:M\to U$
and $G^\pm :  C^\infty_{\pi\,\vpc/\vfc}\otimes\bbK^n \to C^\infty_{\pi\,\vpc/\vfc}\otimes\bbK^n$ 
retarded/advanced Green operators for $P$. 
Then the associated sequence 
\begin{flalign}
\xymatrix@C=1.5em{
0 \ar[r] & C^\infty_{\pi\,\vc}\otimes\bbK^n \ar[r]^-{P} & C^\infty_{\pi\,\vc} \otimes\bbK^n  \ar[r]^-{G} 
& C^\infty_{\pi} \otimes\bbK^n  \ar[r]^-{P} & C^\infty_{\pi} \otimes\bbK^n  \ar[r] & 0
}
\end{flalign}
in $\Sh_{C^\infty_\bbK}(U)$ is exact. Even stronger, the corresponding sequence
of presheaves is exact, i.e.\ for each open subset $U^\prime\subseteq U$, the sequence
\begin{flalign}
\xymatrix@C=1.5em{
0 \ar[r] & C^\infty_{\pi\,\vc}(U^\prime) \otimes\bbK^n  \ar[r]^-{P} &  C^\infty_{\pi\,\vc}(U^\prime) \otimes\bbK^n \ar[r]^-{G} 
& C^\infty_{\pi}(U^\prime) \otimes\bbK^n \ar[r]^-{P} & C^\infty_{\pi}(U^\prime) \otimes\bbK^n \ar[r] & 0
}
\end{flalign}
of $C^\infty_\bbK(U^\prime)$-modules is exact.
\end{propo}
\begin{proof}
We prove the second (stronger) statement, which implies the first.
Let $U^\prime\subseteq U$ be any open subset.
To prove exactness at the first node, consider any $\varphi\in C^\infty_{\pi\,\vc}(U^\prime)\otimes\bbK^n$
such that $P\varphi =0$ and note that $0=G^\pm P \varphi =\varphi$  by Definition \ref{def:Green} (i).
For the second node, let  $\varphi\in C^\infty_{\pi\,\vc}(U^\prime)\otimes\bbK^n$ be such that
$G\varphi =0$. Then $G^+\varphi = G^- \varphi =:\rho \in C^\infty_{\pi\,\vc}(U^\prime)\otimes\bbK^n$ 
because of the support
properties of Green operators and the definition of vertically compact support.
Hence, $P\rho = PG^\pm \varphi =\varphi$ by Definition \ref{def:Green} (i).
\sk

For the third node, let $\Phi\in C^\infty_\pi(U^\prime)\otimes\bbK^n$ be such that $P\Phi=0$.
Choosing two non-intersecting sections $\sigma_\pm : U^\prime \to M\vert_{U^\prime}$
such that $\sigma_+$ lies in the vertical future of $\sigma_-$, we obtain an open cover
$\{M\vert_{U^\prime}\setminus J^-_\ver(\sigma_-(U^\prime) ) ,\,
M\vert_{U^\prime}\setminus J^+_\ver(\sigma_+(U^\prime))\}$ of $M\vert_{U^\prime}$. 
Choosing a partition of unity subordinate to this cover, we can decompose 
$\Phi = \Phi_+ + \Phi_-$ with $\Phi_\pm \in C^\infty_{\pi\, \vpc/\vfc}(U^\prime)\otimes\bbK^n$.
Then $\rho:= P\Phi_+ = -P\Phi_- \in C^\infty_{\pi\, \vc}(U^\prime)\otimes\bbK^n$ 
is vertically compactly supported
and $G\rho = G^+ \rho - G^-\rho = G^+ P\Phi_+ + G^- P\Phi_-= \Phi_+ + \Phi_- =\Phi$ 
by Definition \ref{def:Green} (i).
\sk

For the last node, take any $\Phi\in C^\infty_\pi(U^\prime)\otimes\bbK^n$ and decompose
as before $\Phi= \Phi_+ + \Phi_-$ with $\Phi_\pm \in C^\infty_{\pi\, \vpc/\vfc}(U^\prime)\otimes\bbK^n$.
Defining $\rho:= G^+ \Phi_+ + G^- \Phi_-$, we obtain
$P\rho= PG^+ \Phi_+ + PG^- \Phi_- = \Phi_+ + \Phi_- = \Phi$ by Definition \ref{def:Green} (i).
\end{proof}
\begin{rem}\label{rem:solutions}
As a direct consequence of this proposition, we obtain that the cokernel sheaf
\begin{subequations}
\begin{flalign}
\tfrac{C^\infty_{\pi\,\vc}\otimes\bbK^n}{P(C^\infty_{\pi\,\vc}\otimes\bbK^n)} \,:=\, 
\mathrm{coker} \big( P : C^\infty_{\pi\,\vc}\otimes\bbK^n \to C^\infty_{\pi\,\vc}\otimes\bbK^n\big) \,\in\,\Sh_{C^\infty_{\bbK}}(U)
\end{flalign}
may be computed as a presheaf quotient, i.e.\ 
\begin{flalign}
\tfrac{C^\infty_{\pi\,\vc}\otimes\bbK^n}{P(C^\infty_{\pi\,\vc}\otimes\bbK^n)}(U^\prime) \,=\, 
C^\infty_{\pi\,\vc}(U^\prime)\otimes\bbK^n \big/ P(C^\infty_{\pi\,\vc}(U^\prime) \otimes\bbK^n)\quad,
\end{flalign}
\end{subequations}
for every open subset $U^\prime \subseteq U$. Furthermore, this sheaf is isomorphic
via the causal propagator
\begin{flalign}\label{eqn:causalpropagatoriso}
G \,:\, \xymatrix{
\frac{C^\infty_{\pi\,\vc}\otimes\bbK^n}{P(C^\infty_{\pi\,\vc}\otimes\bbK^n)} \ar[r]^-{\cong} ~&~ \Sol_\pi
}
\end{flalign}
to the solution sheaf $\Sol_\pi := \ker(P : C^\infty_{\pi}\otimes\bbK^n \to C^\infty_{\pi} \otimes\bbK^n)
\in\Sh_{C^\infty_\bbK}(U)$.
\end{rem}

Let us now illustrate these concepts by examples.
\begin{ex}\label{ex:oscillator}
Consider any $U\in \Man$ and any smooth $U$-family of 
$1$-dimensional spacetimes $(\pi:M\to U,E)\in\Loc_1^\infty(U)$.
As equation of motion we take the vertical differential operator
\begin{flalign}\label{eqn:oscillator}
P_{(\pi,E)}\,:=\, {\ast_\ver}  \,{\dd_\ver}  \,{\ast_\ver}\, {\dd_\ver} + m^2 \,:\, C^\infty_{\pi}
~\longrightarrow~C^\infty_\pi\quad,
\end{flalign}
where $m\in (0,\infty)$ is a fixed parameter 
and ${\ast_\ver}  \,{\dd_\ver}  \,{\ast_\ver}\, {\dd_\ver}$ 
is the vertical Laplacian, which is obtained from 
the vertical de Rham differential $\dd_\ver$ 
on $\pi: M\to U$ and the vertical Hodge operator $\ast_\ver$ 
induced by $E\in\Omega^1_{\ver}(M)$. 
This differential operator describes a smooth $U$-family of $1$-dimensional scalar fields
(or equivalently harmonic oscillators) with a fixed mass/frequency parameter $m$
on time intervals whose geometry (i.e.\ length) depends on the point $x\in U$.
To prove that \eqref{eqn:oscillator} admits a retarded and an advanced Green operator,
it is sufficient to prove existence of local retarded and advanced Green operators 
$G^\pm_\alpha  : C^\infty_{\pi\,\vpc/\vfc} \vert_{U_\alpha}\to C^\infty_{\pi\,\vpc/\vfc}\vert_{U_\alpha}$
for an arbitrary choice of open cover $\{U_\alpha\subseteq U\}$.
This is because uniqueness of retarded/advanced Green operators entails
that the family $\{G^\pm_\alpha\}$ satisfies the relevant compatibility conditions
on all overlaps $U_{\alpha\beta}$ and hence, 
recalling from Proposition \ref{prop:Shmodulestack} that $\Sh_{C^\infty_{\bbR}}$ is a stack, 
it defines a global retarded/advanced Green operator.
\sk

To prove local existence, consider any open cover $\{U_\alpha\subseteq U\}$
in which the restricted bundles $M\vert_{U_\alpha}\to U_\alpha$ admit a trivialization
$M\vert_{U_\alpha} \cong \bbR\times U_\alpha$. In this trivialization,
we have that $E\vert_{U_\alpha} \cong \rho\,\dd t $ for a positive
function $\rho\in C^\infty(\bbR\times U_\alpha,\bbR^{>0})$, 
where $t\in \bbR$ is a time coordinate on $\bbR$,
and the equation of motion operator reads as
$P_\alpha = \rho^{-1}\,\partial_t\,\rho^{-1}\partial_t + m^2$.
We can simplify this differential operator even further by introducing
a new ($x\in U_\alpha$ dependent) time coordinate $T(t,x)$ 
such that $\dd_\ver T = \rho \,\dd t $. Note that in these 
coordinates the fiber over $x\in U_\alpha$ is the interval $(T(-\infty,x),T(\infty,x))$,
i.e.\ the geometry/length of the interval may depend on $x$. The equation of
motion operator then reads as $P_\alpha = \partial_{T}^2 +m^2$,
which admits the retarded/advanced Green operator
\begin{flalign}
(G^\pm_\alpha \varphi)(T,x)\,=\, \int_{T(\mp \infty, x)}^{T} m^{-1} \sin\big(m\,(T-S)\big)\, \varphi(S,x)\,\dd S\quad.
\end{flalign}
Because  $\varphi\in C^\infty_{\pi\,\vpc/\vfc}\vert_{U_\alpha}$ is vertically past/future compactly supported,
this integral exists and it depends smoothly 
on both $T\in (T(-\infty,x),T(\infty,x))$ and $x\in U_\alpha$.
\end{ex}

\begin{ex}\label{ex:familyoscillator}
In order to construct in Subsection \ref{subsec:example} 
an example of a smooth $\tilde{U}$-family
of smooth AQFTs, we generalize Example \ref{ex:oscillator}
to the case where the mass parameter is not a constant but rather 
a smooth positive function $m\in C^\infty(\tilde{U},\bbR^{>0})$ on $\tilde{U}\in \Man$.
Given any $U\in \Man$ and  any smooth $U$-family of 
$1$-dimensional spacetimes $(\pi:M\to U,E)\in\Loc_1^\infty(U)$,
we consider the object $(\pi\times \id: M\times \tilde{U}\to U\times\tilde{U},\pr_M^\ast(E))\in\Loc_1^\infty(U\times \tilde{U})$
and define on it the vertical differential operator
\begin{flalign}\label{eqn:familyoscillator}
\tilde{P}_{(\pi,E)}\,:=\, {\ast_\ver}  \,{\dd_\ver}  \, {\ast_\ver}\, {\dd_\ver} + \pr_{\tilde{U}}^\ast(m^2 )\,:\,C^\infty_{\pi\times \id}~\longrightarrow~
C^\infty_{\pi\times \id}\quad,
\end{flalign}
where $\pr_M : M\times \tilde{U}\to M$ and $\pr_{\tilde{U}}:M\times \tilde{U}\to \tilde{U}$ denote the projection maps.
Proceeding in complete analogy to Example \ref{ex:oscillator}, it is sufficient to choose an open cover
$\{U_\alpha\subseteq U\}$ for which there exist bundle trivializations $(M\times \tilde{U})\vert_{U_\alpha\times\tilde{U}} 
\cong \bbR\times U_\alpha\times \tilde{U}$, and prove the existence of Green operators locally.
Choosing again a convenient time coordinate $T$ by solving $\dd_\ver T =\rho\,\dd t\cong \pr_M^\ast(E)\vert_{U_\alpha\times\tilde{U}}$, the local
equation of motion operator reads as $\tilde{P}_\alpha = \partial_{T}^2 + m^2(\tilde{x})$,
where we made the dependence on $\tilde{x}\in \tilde{U}$ explicit. This operator admits
a retarded/advanced Green operator given by
\begin{flalign}
(\tilde{G}^\pm_\alpha \varphi)(T,x,\tilde{x})\,=\, \int_{T(\mp \infty, x)}^{T} 
m(\tilde{x})^{-1} \sin\big(m(\tilde{x})\,(T-S)\big)\, \varphi(S,x,\tilde{x})\,\dd S\quad,
\end{flalign}
for all $\varphi \in C^\infty_{\pi\times\id\,\vpc/\vfc}\vert_{U_\alpha\times \tilde{U}}$.
\end{ex}

\begin{ex}\label{ex:Dirac}
Consider any $U\in \Man$ and any smooth $U$-family of 
$1$-dimensional spacetimes $(\pi:M\to U,E)\in\Loc_1^\infty(U)$.
The $1$-dimensional massless Dirac field is described by the vertical differential operator
\begin{flalign}\label{eqn:Diracoperator}
D_{(\pi,E)}\,:=\, \begin{pmatrix}
\ii\,{\ast_\ver}  \,{\dd_\ver}  & 0\\
0& -\ii \,{\ast_\ver}  \,{\dd_\ver}  
\end{pmatrix}\,:\, C^\infty_{\pi}\otimes\bbC^2
~\longrightarrow~C^\infty_\pi\otimes\bbC^2\quad,
\end{flalign}
where $\ii\in\bbC$ is the imaginary unit and the
elements $\big(\mycom{\Psi}{\overline{\Psi}}\big)\in C^\infty_{\pi}\otimes\bbC^2$
should be interpreted as the Dirac field $\Psi$ and its Dirac conjugate
$\overline{\Psi}$. The existence of retarded/advanced Green operators
can be proven as in the previous examples by a local argument.
Indeed, restricting again to a trivializing cover $\{U_\alpha\subseteq U\}$ and introducing the
local time coordinate $T$, the local Dirac operator reads as
\begin{flalign}
D_\alpha\,=\,\begin{pmatrix}
\ii\partial_T^{} & 0\\
0& -\ii\partial_T^{} 
\end{pmatrix}
\end{flalign}
and its associated retarded/advanced Green operator is given by fiber integration
\begin{flalign}
\Big(S^\pm_\alpha\Big({\footnotesize \mycom{\psi}{\overline{\psi}}}\Big)\Big)(T,x)
\,=\,\int_{T(\mp \infty, x)}^{T}\begin{pmatrix}
-\ii\,\psi(S,x)\\
\ii\,\overline{\psi}(S,x)
\end{pmatrix}~\dd S\quad,
\end{flalign}
where $T(\mp \infty, x)$ was defined in Example \ref{ex:oscillator}.
\end{ex}

We conclude this subsection with a few remarks about smoothly parametrized
initial value problems. Let us start with the case where the $\Sh_{C^\infty_\bbK}(U)$-morphism
$P : C^\infty_{\pi}\otimes\bbK^n\to C^\infty_{\pi}\otimes\bbK^n$ 
corresponds to a second order vertical differential
operator on $\pi:M\to U$, as it is the case in our Examples \ref{ex:oscillator} and \ref{ex:familyoscillator}. 
Choosing any section $\sigma : U\to M$, we can define a $\Sh_{C^\infty_\bbK}(U)$-morphism
\begin{subequations}\label{eqn:datamap}
\begin{flalign}
\data^{\mathrm{2nd}}_{\sigma}\,:\, \Sol_\pi~\longrightarrow~\big(C^\infty_{\bbR,U}\otimes\bbK^n\big)^{\oplus 2}
\end{flalign}
that assigns, for each open subset $U^\prime\subseteq U$, 
to a solution $\Phi\in \Sol_\pi(U^\prime)\subseteq C^\infty_\pi(U^\prime)\otimes\bbK^n 
=C^\infty_\bbR(M\vert_{U^\prime})\otimes\bbK^n$ 
(see Remark \ref{rem:solutions}) its initial data
\begin{flalign}
\data^{\mathrm{2nd}}_{\sigma}(\Phi) \,:=\, \big(\sigma^\ast(\Phi), \sigma^\ast({\ast_\ver}{\dd_\ver} \Phi)\big)\,\in\, 
\big(C^\infty_{\bbR}(U^\prime)\otimes\bbK^n\big)^{\oplus 2}
\end{flalign}
\end{subequations}
on $\sigma(U^\prime)\subseteq M\vert_{U^\prime}$.
We say that $P$ has a well-posed initial value problem if \eqref{eqn:datamap} 
is an isomorphism in $\Sh_{C^\infty_\bbK}(U)$. 
Note that if $P$ has retarded/advanced Green operators and a well-posed initial value problem,
it follows by using also Remark \ref{rem:solutions} that
\begin{flalign}
\xymatrix{
\tfrac{C^\infty_{\pi\,\vc}\otimes\bbK^n}{P(C^\infty_{\pi\,\vc}\otimes\bbK^n)} \ar[r]^-{G}~&~ \Sol_{\pi} \ar[r]^-{\data^{\mathrm{2nd}}_{\sigma}} ~&~\big(C^\infty_{\bbR,U}\otimes\bbK^n\big)^{\oplus 2}
}
\end{flalign}
is an isomorphism in $\Sh_{C^\infty_\bbK}(U)$ and hence that
$C^\infty_{\pi\,\vc}\otimes\bbK^n/P(C^\infty_{\pi\,\vc}\otimes\bbK^n) \in\Sh_{C^\infty_\bbK}(U) $
is a free $C^\infty_{\bbK, U}$-module of rank $2n$. This observation
will be useful in Subsection \ref{subsec:example}.
\begin{ex}\label{ex:datamap}
The equation of motion operators in Examples \ref{ex:oscillator} 
and \ref{ex:familyoscillator} have a well-posed initial value problem. 
Let us show this for the more general operator \eqref{eqn:familyoscillator} in the latter example, 
which reduces for $\tilde{U}=\{\ast\}$ a point to the operator \eqref{eqn:oscillator}.
Using again that $\Sh_{C^\infty_\bbR}$ is a stack, it is sufficient
to prove the isomorphism property of the sheaf morphism
\eqref{eqn:datamap} in each patch $U_\alpha$ of an arbitrary
open cover $\{U_\alpha\subseteq U\}$. Using as in Example \ref{ex:familyoscillator} 
a trivializing cover and suitable time coordinates $T$, we obtain the local equation 
of motion operator $\tilde{P}_\alpha =\partial_{T}^2 + m^2(\tilde{x})$.
The inverse of the restriction of the initial data map $\data^{\mathrm{2nd}}_\sigma$ to $U_\alpha\subseteq U$ 
is then given by
\begin{multline}
\solve_{\sigma}\big(\Phi_0,\Phi_1\big)(T,x,\tilde{x})\,:=\,
\Phi_0(x,\tilde{x}) \,\cos\big(m(\tilde{x})\, (T-T_\sigma(x))\big) \\
+ \Phi_1(x,\tilde{x})\, m(\tilde{x})^{-1} \,\sin\big(m(\tilde{x})\, (T-T_\sigma(x))\big)\quad,
\end{multline}
for all $(\Phi_0,\Phi_1)\in (C^\infty_{\bbR,U_\alpha\times \tilde{U}})^{\oplus 2}$,
where the initial time $T_\sigma(x)\in (T(-\infty,x),T(\infty,x))$ 
is determined from the local coordinate expression $\sigma(x) = (T_\sigma(x),x)$, for all $x\in U_\alpha$, 
of the section $\sigma$.
\end{ex}

\begin{rem}
The case of a first order vertical differential
operator $P : C^\infty_{\pi}\otimes\bbK^n\to C^\infty_{\pi}\otimes\bbK^n$ 
on $\pi:M\to U$ works similarly.
The analog of \eqref{eqn:datamap} is given by the $\Sh_{C^\infty_\bbK}(U)$-morphism
\begin{subequations}\label{eqn:1storderdata}
\begin{flalign}
\data^{\mathrm{1st}}_{\sigma}\,:\, \Sol_\pi~\longrightarrow~C^\infty_{\bbR,U}\otimes\bbK^n
\end{flalign}
that assigns, for each open subset $U^\prime\subseteq U$, 
to a solution $\Phi\in \Sol_\pi(U^\prime)$ its initial data
\begin{flalign}
\data^{\mathrm{1st}}_{\sigma}(\Phi) \,:=\, \sigma^\ast(\Phi)\,\in\, 
C^\infty_{\bbR}(U^\prime)\otimes\bbK^n
\end{flalign}
\end{subequations}
on $\sigma(U^\prime)\subseteq M\vert_{U^\prime}$. We again say that
$P$ has a well-posed initial value problem if \eqref{eqn:1storderdata}
is an isomorphism. It is easy to check that the Dirac operator 
from Example \ref{ex:Dirac} has a well-posed initial value problem
in this sense. 
\end{rem}

\subsection{\label{subsec:example}$1$-dimensional scalar field}
The aim of this subsection is to construct an explicit
smooth $\tilde{U}$-family of smooth AQFTs that can be interpreted
as a smooth refinement of the $1$-dimensional scalar field with a 
smoothly varying mass parameter $m\in C^\infty(\tilde{U},\bbR^{>0})$. 
Our construction will be based on Example \ref{ex:familyoscillator}
and we will carry out (smooth generalizations of) 
the usual steps in the construction of Bosonic free field theories, see e.g.\ \cite{BG,BGP}.
Since we are interested in smooth $\tilde{U}$-families (see \eqref{eqn:tildeUfamilyAQFT}),
we have to define instead of \eqref{eqn:AAAviaCCR} a stack morphism
\begin{flalign}\label{eqn:BBBviaCCR}
\xymatrix{
\ar[dr]_-{\WWW}\Loc_1^\infty \ar[rr]^-{\BBB}~&~~&~\Map\big(\und{\tilde{U}},\astAlg^\infty\big)\\
~&~ \Map\big(\und{\tilde{U}},\Pois^\infty\big) \ar[ur]_-{~~\Map(\und{\tilde{U}},\CCR)}~&~
}\quad,
\end{flalign}
which for the case of $\tilde{U}=\{\ast\}$ a point reduces to \eqref{eqn:AAAviaCCR}.
For our example of interest, the 
stack morphism $\WWW:\Loc_1^\infty \to \Map(\und{\tilde{U}},\Pois^\infty)$ is given
by the following data: For each manifold $U\in\Man$, we define the functor
\begin{flalign}\label{eqn:EWWUfunctor}
\WWW_U\,:\,\Loc^\infty_1(U)~\longrightarrow~\Pois^\infty(U\times \tilde{U})
\end{flalign}
that assigns to $(\pi:M\to U,E)\in \Loc_1^\infty(U)$ the object
\begin{subequations}
\begin{flalign}\label{eqn:WWWexplicit}
\WWW_U\big(\pi:M\to U,E\big) \,:=\,
\Big(\tfrac{C^\infty_{\pi\times\id\,\vc}}{\tilde{P}_{(\pi,E)} C^\infty_{\pi\times\id\,\vc}} , \tau_{(\pi\times\id,\pr_M^\ast(E))}\Big)\,\in \,\Pois^\infty(U\times \tilde{U})\quad,
\end{flalign}
where $\tilde{P}_{(\pi,E)}$ is the equation of motion operator in \eqref{eqn:familyoscillator}.
The Poisson structure reads as
\begin{flalign}\label{eqn:WWWpoisson}
 \tau_{(\pi\times\id,\pr_M^\ast(E))}\,=\, \langle\cdot ,\tilde{G}_{(\pi,E)}(\cdot)\rangle_{(\pi\times\id,\pr_M^\ast(E) )}\quad,
\end{flalign}
where $\tilde{G}_{(\pi,E)}$ is the causal propagator for $\tilde{P}_{(\pi,E)}$,
whose existence was established in Example \ref{ex:familyoscillator},
and 
\begin{flalign}\label{eqn:WWWpairing}
\langle \cdot ,\cdot \rangle_{(\pi\times\id,\pr_M^\ast(E) )} \,:\, \xymatrix{
C^\infty_{\pi\times\id\,\vc} \otimes_{C^\infty_{\bbR,U\times\tilde{U}}} C^\infty_{\pi\times\id} \ar[r]^-{\mu}~&~
C^\infty_{\pi\times\id\,\vc} \ar[rr]^-{\int_{\pi\times\id}(-)\pr_M^\ast(E)}~&&~C^\infty_{\bbR,U\times\tilde{U}}
}\quad.
\end{flalign}
\end{subequations}
is the $\Sh_{C^\infty_\bbR}(U\times\tilde{U})$-morphism
obtained by composing the multiplication map $\mu$ of functions
and fiber integration on $(\pi\times\id :M\times\tilde{U}\to U\times\tilde{U},\pr_M^\ast(E))
\in\Loc^\infty(U\times\tilde{U})$. 
Less formally, we can write for \eqref{eqn:WWWpoisson} also
the more familiar looking expression
\begin{flalign}
 \tau_{(\pi\times\id,\pr_M^\ast(E))}\big(\varphi,\varphi^\prime\big)\,=\,\int_\pi \varphi 
 \,\tilde{G}_{(\pi,E)}(\varphi^\prime)\,\pr_M^\ast(E)\quad.
\end{flalign}
Note that the Poisson structure
is well-defined on the quotient in \eqref{eqn:WWWexplicit} because
$\tilde{P}_{(\pi,E)}$ is formally self-adjoint with respect to the pairing \eqref{eqn:WWWpairing}
and $\tilde{G}_{(\pi,E)}\circ \tilde{P}_{(\pi,E)} = 0=
\tilde{P}_{(\pi,E)}\circ \tilde{G}_{(\pi,E)}$ due to the definition of Green operators,
see Definition \ref{def:Green}.
\sk

Given any morphism $f : (\pi:M\to U,E)\to (\pi^\prime:M^\prime\to U,E^\prime)$
in $ \Loc_1^\infty(U)$, we can define a morphism
$f\times \id  :  (\pi\times \id :M\times\tilde{U}\to U\times \tilde{U},\pr_M^\ast(E))\to (\pi^\prime\times\id :
M^\prime\times\tilde{U}\to U\times\tilde{U},\pr_{M^\prime}^\ast(E^\prime))$ in
$\Loc^\infty_1(U\times\tilde{U})$. This yields a pushforward
(i.e.\ extension by zero) $\Sh_{C^\infty_\bbR}(U\times\tilde{U})$-morphism
$(f\times\id)_\ast : C^\infty_{\pi\times\id\,\vc}\to C^\infty_{\pi^\prime\times\id\,\vc}$
of vertically compactly supported functions (recall from Definition \ref{def:Loc_1(U)}
that $f\times \id$ is an open embedding of fiber bundles) that intertwines the equation of motion operators,
i.e.\ $\tilde{P}_{(\pi^\prime,E^\prime)}\, (f\times\id)_\ast = (f\times\id)_\ast\, \tilde{P}_{(\pi,E)}$.
We can then define the values of \eqref{eqn:EWWUfunctor} on morphisms by
\begin{flalign}\label{eqn:WWWUmorphisms}
\WWW_U(f) \,:=\, (f\times\id)_\ast \,:\, \WWW_U\big(\pi:M\to U,E\big)~\longrightarrow~
\WWW_U\big(\pi^\prime:M^\prime\to U,E^\prime\big)\quad.
\end{flalign}
Note that the preservation of Poisson structures follows from the uniqueness of retarded/advanced 
Green operators.
\sk

To complete the definition of the stack morphism $\WWW$,
it remains to provide, for each morphism $h:U\to U^\prime$ in $\Man$,
a natural isomorphism (see Definition \ref{def:stackmorphism})
\begin{subequations}\label{eq:Wh}
\begin{flalign}
\xymatrix@C=4em{
\ar[d]_-{h^\ast} \Loc_1^\infty(U^\prime) \ar[r]^-{\WWW_{U^\prime}}~&~\Pois^\infty(U^\prime\times\tilde{U})\ar[d]^-{(h \times \id)^\ast} \ar@{=>}[dl]_-{\WWW_h}\\
\Loc_1^\infty(U) \ar[r]_-{\WWW_U}~&~\Pois^\infty(U\times\tilde{U})
}
\end{flalign}
To define the component
\begin{flalign}\label{eq:WWWh}
\WWW_h\,:\, (h \times \id)^\ast \WWW_{U^\prime}\big(\pi:M\to U^\prime,E\big)
~\longrightarrow~ \WWW_U \big(h^\ast (\pi: M\to U^\prime,E) \big)
\end{flalign}
\end{subequations}
at $(\pi:M\to U^\prime,E) \in \Loc_1^\infty(U^\prime)$,
we recall the pullback bundle construction in \eqref{eqn:pullbackbundle} and 
\eqref{eqn:Locfunctor2} and consider the $\Sh_{C^\infty_\bbR}(U\times\tilde{U})$-morphism 
\begin{flalign}\label{eqn:overlinehpullback1}
(\overline{h}^M\times\id)^\ast \,:\, (h\times\id)^\ast C^\infty_{\pi\times\id\,\vc}~\longrightarrow~
C^\infty_{\pi_h\times\id\,\vc}
\end{flalign}
that is defined through its adjunct under 
$(h\times\id)^\ast : \Sh_{C^\infty_\bbR}(U\times\tilde{U}) \rightleftarrows \Sh_{C^\infty_\bbR}(U^\prime\times\tilde{U})$
by the components (denoted with abuse of notation by the same symbol)
\begin{flalign}\label{eqn:overlinehpullback2}
(\overline{h}^{M}\times\id)^{\ast}\,:\,C^\infty_\bbR\big((M\times \tilde{U})\vert_{U^{\prime\prime}}\big) ~\longrightarrow~
C^\infty_\bbR\big((h^\ast M\times \tilde{U})\vert_{(h\times\id)^{-1}(U^{\prime\prime})}\big) \quad,
\end{flalign}
for all open subsets $U^{\prime\prime}\subseteq U^\prime\times\tilde{U}$, 
which describe the pullback of functions along the map of total spaces. Due to the universal property of 
pullback bundles, one easily checks that each section $\sigma : U^\prime\times \tilde{U}\to M\times\tilde{U}$
induces a section $\sigma_h : U\times \tilde{U}\to h^\ast M\times\tilde{U}$ of the pullback bundle
that satisfies $\sigma\,(h\times\id) = (\overline{h}^M\times\id)\, \sigma_h$, hence
the maps in \eqref{eqn:overlinehpullback2} preserve vertically compact support.
Due to naturality of the vertical differential operators $\tilde{P}$ in \eqref{eqn:familyoscillator},
we obtain the commutative diagram
\begin{flalign}\label{eqn:overlinehpullbackP}
\xymatrix@C=4em{
(h \times \id)^\ast C^\infty_{\pi\times\id\,\vc} \ar[d]_-{(h \times \id)^\ast \tilde{P}_{(\pi,E)}} \ar[r]^-{(\ovr{h}^{M} \times \id)^\ast} ~&~ C^\infty_{\pi_h\times\id\,\vc} \ar[d]^-{\tilde{P}_{h^\ast(\pi,E)}}\\
(h \times \id)^\ast C^\infty_{\pi\times\id\,\vc} \ar[r]_-{(\ovr{h}^{M} \times \id)^\ast} ~&~ C^\infty_{\pi_h\times\id\,\vc}
}
\end{flalign}
in $\Sh_{C^\infty_\bbR}(U\times\tilde{U})$, which allows us to induce \eqref{eqn:overlinehpullback1} 
to the quotients
\begin{flalign}\label{eqn:overlinehpullback3} 
(\overline{h}^{M}\times\id)^{\ast} \,:\,  (h\times\id)^\ast \Big(\tfrac{C^\infty_{\pi\times\id\,\vc}}{ \tilde{P}_{(\pi,E)} C^\infty_{\pi\times\id\,\vc}}\Big)
~\longrightarrow~\tfrac{C^\infty_{\pi_h\times\id\,\vc}}{\tilde{P}_{h^\ast(\pi,E)} C^\infty_{\pi_h\times\id\,\vc}}\quad.
\end{flalign}
Here we also used that $(h\times\id)^\ast $ is a left adjoint functor, hence it commutes with the colimit
defining these quotients. From the explicit expression \eqref{eqn:overlinehpullback2} 
for (the adjunct of) this morphism and observing that a diagram similar to \eqref{eqn:overlinehpullbackP} 
involving retarded/advanced Green operators commutes due to their uniqueness,
one checks that \eqref{eqn:overlinehpullback3} preserves the relevant Poisson structures
and thereby defines the desired $\Pois^\infty(U^\prime\times \tilde{U})$-morphism
$\WWW_h := (\overline{h}^{M}\times\id)^{\ast}$ in \eqref{eq:WWWh}.
\sk

It remains to confirm that \eqref{eqn:overlinehpullback3} is an 
isomorphism in $\Pois^\infty(U^\prime\times \tilde{U})$. Using the 
causal propagators \eqref{eqn:causalpropagatoriso} and the initial
data morphisms \eqref{eqn:datamap} corresponding to any choice of
section $\sigma : U^\prime\times \tilde{U}\to M\times\tilde{U}$
and its induced section $\sigma_h : U\times \tilde{U}\to h^\ast M\times\tilde{U}$ of the pullback bundle,
we obtain the commutative diagram
\begin{flalign}\label{eqn:Gdatadiagram}
\xymatrix@C=4em{
\ar[d]_-{(h\times\id)^\ast \tilde{G}_{(\pi,E)}} (h\times\id)^\ast \Big(\tfrac{C^\infty_{\pi\times\id\,\vc}}{ \tilde{P}_{(\pi,E)} C^\infty_{\pi\times\id\,\vc}}\Big) \ar[r]^-{(\overline{h}^{M}\times\id)^{\ast}}~&~ \tfrac{C^\infty_{\pi_h\times\id\,\vc}}{\tilde{P}_{h^\ast(\pi,E)} C^\infty_{\pi_h\times\id\,\vc}}
\ar[d]^-{ \tilde{G}_{h^\ast(\pi,E)}}\\
\ar[d]_-{(h\times\id)^\ast\data^{\mathrm{2nd}}_\sigma}(h\times\id)^\ast \Sol_{\pi\times \id} \ar[r]^-{(\overline{h}^{M}\times\id)^{\ast}}~&~  \Sol_{\pi_h\times \id} 
\ar[d]^-{\data^{\mathrm{2nd}}_{\sigma_h}}\\
(h\times\id)^\ast \big(C^\infty_{\bbR,U^\prime \times \tilde{U}}\big)^{\oplus 2} \ar[r]_-{\cong}~&~ \big(C^\infty_{\bbR,U \times \tilde{U}}\big)^{\oplus 2} 
}
\end{flalign}
in $\Sh_{C^\infty_\bbR}(U \times \tilde{U})$, where the bottom horizontal isomorphism
uses that $(h\times\id)^\ast$ preserves coproducts (as it is a left adjoint functor)
and the symmetric monoidal coherence isomorphism for the monoidal unit 
in \eqref{eqn:ShSMcoherences}. By Remark \ref{rem:solutions}
and Examples \ref{ex:familyoscillator} and \ref{ex:datamap}, 
all vertical arrows in this diagram are isomorphisms, hence the top 
horizontal arrow is an isomorphism too. This implies that $\WWW_h$ is an 
isomorphism in $\Pois^\infty(U \times \tilde{U})$. 
\sk

Summing up, the main result of this section is
\begin{propo}
The construction described above defines a stack morphism
$\WWW:\Loc_1^\infty \to \Map(\und{\tilde{U}},\Pois^\infty)$. 
\end{propo}

As a consequence, we obtain an explicit example of a smooth $\tilde{U}$-family 
of smooth $1$-dimensional AQFTs
$\BBB := \Map(\und{\tilde{U}},\CCR) \circ \WWW: \Loc_1^\infty \to \Map(\und{\tilde{U}},\astAlg^\infty)$
describing a smooth refinement of the $1$-dimensional scalar field with
a smoothly varying mass parameter $m\in C^\infty(\tilde{U},\bbR^{>0})$. 
In the special case where $\tilde{U}=\{\ast\}$ is a point, our construction
describes a smooth refinement of the $1$-dimensional scalar field
with a fixed mass $m>0$.

\subsection{\label{subsec:fermion}$1$-dimensional Dirac field}
The construction of Subsection \ref{subsec:example} can be 
easily adapted to the case of the $1$-dimensional Dirac field 
introduced in Example \ref{ex:Dirac}. 
We will spell out the relevant steps to construct
the corresponding stack morphism 
$\LLL^{\mathsf{f}} : \Loc_1^\infty\to\IPVec^\infty$
such that $\AAA^{\mathsf{f}}:= \CAR\circ\LLL^{\mathsf{f}}$
in \eqref{eqn:AAAviaCAR} describes a smooth refinement
of the massless $1$-dimensional Dirac field. For this we will
carry out (smooth generalizations of) 
the usual steps in the construction of Fermionic 
free field theories, see e.g.\ \cite{Dappiaggi}. 
After that we will show
that the smooth automorphism group \eqref{eqn:Aut} 
of this model includes the global $U(1)$-symmetry of the Dirac field.
\sk

The stack morphism $\LLL^{\mathsf{f}} : \Loc_1^\infty\to\IPVec^\infty$
is given by the following data: For each manifold $U\in \Man$,
we define the functor
\begin{flalign}\label{eqn:LLLDiracUfunctor}
\LLL^{\mathsf{f}}_U\,:\, \Loc_1^\infty(U)~\longrightarrow~\IPVec^\infty(U)
\end{flalign}
that assigns to each $(\pi: M\to U,E)\in\Loc_1^\infty(U)$ the object
\begin{subequations}\label{eqn:LLLDiracAll}
\begin{flalign}\label{eqn:LLLDirac}
\LLL^{\mathsf{f}}_U\big(\pi:M\to U,E\big) \,:=\,
\Big(\tfrac{C^\infty_{\pi\,\vc}\otimes\bbC^2}{D_{(\pi,E)} (C^\infty_{\pi\,\vc}\otimes\bbC^2)} 
, \ast_{(\pi,E)}, \langle\cdot,\cdot\rangle_{(\pi,E)}\Big)\,\in \,\IPVec^\infty(U)\quad,
\end{flalign}
where $D_{(\pi,E)}$ is the Dirac operator from Example \ref{ex:Dirac}.
The $\ast$-involution
\begin{flalign}\label{eqn:LLLinvolution}
{\ast_{(\pi,E)}}\begin{pmatrix}
\psi\\\overline{\psi}
\end{pmatrix}\,:=\, 
\begin{pmatrix}
\overline{\psi}^\ast\\
\psi^\ast
\end{pmatrix}
\end{flalign}
is given by swapping the components followed by complex conjugation,
which descends to the quotient since
${\ast_{(\pi,E)}}\circ D_{(\pi,E)} =D_{(\pi,E)} \circ {\ast_{(\pi,E)}}$.
The symmetric pairing 
\begin{flalign}\label{eqn:LLLpairing}
\Big\langle\begin{pmatrix}
\psi\\\overline{\psi}
\end{pmatrix}
,\begin{pmatrix}
\psi^\prime\\\overline{\psi}^\prime
\end{pmatrix}
\Big\rangle_{(\pi,E)}\,:=\, \int_{\pi} \begin{pmatrix}
\psi & \overline{\psi} 
\end{pmatrix}\, \begin{pmatrix}
0&\ii\\
-\ii & 0
\end{pmatrix} S_{(\pi,E)} \begin{pmatrix}
\psi^\prime\\\overline{\psi}^\prime
\end{pmatrix}\, E
\end{flalign}
\end{subequations}
is given by fiber integration, the causal propagator $S_{(\pi,E)}$ for
$D_{(\pi,E)}$ and the displayed matrix multiplications. It is easy
to check that $ \langle\cdot,\cdot\rangle_{(\pi,E)}$ descends to the 
quotient in \eqref{eqn:LLLDirac} and that it satisfies
the compatibility condition \eqref{eqn:astFermionCompatible} 
for $\ast$-involutions.
The definition of the functor \eqref{eqn:LLLDiracUfunctor}
on morphisms $f : (\pi : M\to U,E)\to (\pi^\prime:M^\prime\to U,E^\prime)$
is as in \eqref{eqn:WWWUmorphisms} via pushforward of vertically compactly supported functions.
The coherence isomorphisms for $\Man$-morphisms $h:U\to U^\prime$
are constructed in complete analogy to \eqref{eq:Wh}.
\sk

Summing up, we obtain
\begin{propo}
The construction described above defines a stack morphism
$\LLL^{\mathsf{f}}:\Loc_1^\infty \to \IPVec^\infty$. 
\end{propo}

As a consequence, we obtain another example of a
smooth $1$-dimensional AQFT $\AAA^{\mathsf{f}} := \CAR \circ \LLL^{\mathsf{f}}
: \Loc_1^\infty \to \astAlg^\infty$
describing a smooth refinement of the $1$-dimensional massless Dirac field.
\sk

To conclude this section, we will show that our construction
can be refined to define a $U(1)$-equivariant smooth AQFT, showing that
the global $U(1)$-symmetry of the Dirac field is smooth in our sense.
Recalling from \eqref{eqn:equivariantAQFTexplicit} their explicit description, 
we will define a $U(1)$-equivariant smooth AQFT 
\begin{flalign}
\xymatrix{
\ar[dr]_-{\tilde{\LLL}^{\mathsf{f}}}\Loc_1^\infty\times [\{\ast\}/U(1)]_{\mathrm{pre}} 
\ar[rr]^-{\tilde{\AAA}^{\mathsf{f}}}~&~~&~\astAlg^\infty\\
~&~ \IPVec^\infty \ar[ur]_-{\CAR}~&~
}
\end{flalign}
by specifying a pseudo-natural transformation $\tilde{\LLL}^{\mathsf{f}}$.
For each manifold $U\in \Man$, we define the functor
\begin{flalign}
\tilde{\LLL}^{\mathsf{f}}_U\,:\, \Loc_1^\infty(U)\times [\{\ast\}/U(1)]_{\mathrm{pre}}(U)~\longrightarrow~\IPVec^\infty(U)
\end{flalign}
that acts on objects $(\pi: M\to U,E)\in\Loc_1^\infty(U)\times[\{\ast\}/U(1)]_{\mathrm{pre}}(U)$ 
precisely as in \eqref{eqn:LLLDiracAll}, i.e.\
\begin{flalign}\label{eqn:LLLDiracequv}
\tilde{\LLL}^{\mathsf{f}}_U\big(\pi:M\to U,E\big)\,:=\, \LLL^{\mathsf{f}}_U\big(\pi:M\to U,E\big)\,=\,
\Big(\tfrac{C^\infty_{\pi\,\vc}\otimes\bbC^2}{D_{(\pi,E)} (C^\infty_{\pi\,\vc}\otimes\bbC^2)} 
, \ast_{(\pi,E)}, \langle\cdot,\cdot\rangle_{(\pi,E)}\Big)\quad.
\end{flalign} 
(Note that the objects of $\Loc_1^\infty(U)\times[\{\ast\}/U(1)]_{\mathrm{pre}}(U)$ 
are canonically identified with the objects of $\Loc_1^\infty(U)$ because $[\{\ast\}/U(1)]_{\mathrm{pre}}(U)$
has only a single object, see \eqref{eqn:quotientprestack}.)
Things get more interesting at the level of morphisms,
because the morphisms in $\Loc_1^\infty(U)\times [\{\ast\}/U(1)]_{\mathrm{pre}}(U)$
are pairs $(f,g)$ consisting of a $\Loc_1^\infty(U)$-morphism
$f : (\pi : M\to U,E)\to (\pi^\prime:M^\prime\to U,E^\prime)$ 
and a $U(1)$-valued smooth function $g\in C^\infty(U,U(1))$.
These morphisms are defined to act 
by a combination of the pushforward of vertically compactly 
supported functions and a complex phase rotation
\begin{flalign}\label{eqn:groupactionDirac}
\tilde{\LLL}^{\mathsf{f}}_U(f,g) \begin{pmatrix}
\psi\\\overline{\psi}
\end{pmatrix}
\,:=\,  \begin{pmatrix}
f_\ast\big(\pi^\ast(g)\,\psi\big)\\ f_\ast\big(\pi^\ast(g)^{-1}\, \overline{\psi}\big)
\end{pmatrix}
\,=\,  \begin{pmatrix}
\pi^\ast(g)\,f_\ast(\psi)\\ \pi^\ast(g)^{-1}\, f_\ast(\overline{\psi})
\end{pmatrix}
\quad,
\end{flalign}
where $\pi^\ast(g)$ denotes the pullback of $g\in C^\infty(U,U(1))$ 
along the projection map $\pi:M\to U$. (The second equality
in \eqref{eqn:groupactionDirac} follows from the fact that 
$f_\ast$ only acts along the fibers where $\pi^\ast(g)$ is constant.)
These maps clearly preserve the quotient in \eqref{eqn:LLLDiracequv}, 
the $\ast$-involution \eqref{eqn:LLLinvolution}
and the pairing \eqref{eqn:LLLpairing}, hence
they define $\IPVec^\infty(U)$-morphisms.
The coherence isomorphisms for $\Man$-morphisms $h:U\to U^\prime$
are constructed in complete analogy to our previous examples.
\sk

Summing up, we obtain
\begin{propo}
The construction described above defines a pseudo-natural transformation
$\tilde{\LLL}^{\mathsf{f}}:\Loc_1^\infty\times [\{\ast\}/U(1)]_{\mathrm{pre}} \to \IPVec^\infty$. 
\end{propo}

As a consequence, we obtain an example of a $U(1)$-equivariant
smooth $1$-dimensional AQFT $\tilde{\AAA}^{\mathsf{f}} := \CAR \circ \tilde{\LLL}^{\mathsf{f}}
: \Loc_1^\infty\times  [\{\ast\}/U(1)]_{\mathrm{pre}} \to \astAlg^\infty$
describing a smooth refinement of the $1$-dimensional massless Dirac field
together with its global $U(1)$-symmetry.

%%%%%%%%%%%%%%%%%%%%%%%%%%%%%%%%%%%%%%%%%%%%%%%%

\section{\label{sec:higher}Outlook: Towards higher dimensions and gauge theories}
The aim of this section is to outline the way we believe 
the results of this paper could be generalized to higher-dimensional AQFTs and also
to gauge theories. In particular, we shall explain 
the additional technical challenges and open questions that arise from such a generalization.
\sk

Let us first discuss possible generalizations
of the stack $\Loc_1^\infty$ of $1$-dimensional spacetimes
from Subsection \ref{subsec:Locstack} to the case of higher dimensions $m\geq 2$.
For $U\in\Man$ a manifold, we can define a smooth $U$-family
of $m$-dimensional Lorentzian manifolds to be a tuple $(\pi:M\to U, g)$
consisting of a (locally trivializable) fiber bundle  $\pi:M\to U$ with typical fiber an
$m$-manifold $N$ and a metric $g$ of signature $(+-\cdots -)$ on the vertical tangent 
bundle of $\pi:M\to U$. For illustrative purposes, we note that in a local trivialization
$M\vert_{U^\prime} \cong N\times U^\prime$ and in local coordinates 
$y^\mu$ on the fiber $N$, the vertical metric takes the form $g\vert_{U^\prime} \cong g_{\mu\nu}(y,x)\,
\dd y^\mu\otimes \dd y^\nu $, i.e.\ it has only vertical components along $N$
that however are allowed to depend smoothly on $x\in U^\prime\subseteq U$.
There are obvious notions of vertical orientation $\mathfrak{o}$ and vertical time-orientation
$\mathfrak{t}$, hence we can introduce a concept of smooth $U$-families 
of $m$-dimensional oriented and time-oriented Lorentzian manifolds
$(\pi:M\to U,g,\mathfrak{o},\mathfrak{t})$.
What is less obvious is the correct generalization of the important concept of
global hyperbolicity to this smoothly parametrized context. One could either 
impose the point-wise condition that each fiber $(M\vert_x,g\vert_x)$ is globally hyperbolic
in the usual sense or seek for a condition that is more uniform on $U$. The role of
this condition should be to ensure that vertical normally hyperbolic 
operators, such as the vertical Klein-Gordon operator
\begin{flalign}\label{eqn:KGoperator}
P\,:=\, {\ast_\ver}{\dd_\ver}{\ast_\ver}{\dd_\ver} + m^2 \,:\ C^\infty_\pi ~\longrightarrow~C^\infty_\pi\quad,
\end{flalign}
admit retarded and advanced Green operators and a well-posed initial value problem, 
both described in terms of morphisms of sheaves of $C^\infty_{\bbR, U}$-modules. 
This can be interpreted saying that both the retarded/advanced Green operators and 
the initial value problem are smoothly parametrized.
Again for illustrative purposes, we note that in a local trivialization
$M\vert_{U^\prime} \cong N\times U^\prime$ and in local coordinates 
$y^\mu$ on the fiber $N$, the vertical differential operator \eqref{eqn:KGoperator}
takes the form
\begin{flalign}
P\vert_{U^\prime}\,\cong\, g^{\mu\nu}(y,x)\frac{\partial^2}{\partial y^\mu\partial y^\nu} 
+ B^{\mu}(y,x)\frac{\partial }{\partial y^\mu} + A(y,x)\quad,
\end{flalign}
i.e.\ there are no derivatives along $x\in U^\prime$ but the coefficients may be $x$-dependent.
Summing up, we record
\begin{problem}\label{problem1}
Find a suitable generalization of global hyperbolicity to smooth $U$-families 
of $m$-dimensional oriented and time-oriented Lorentzian manifolds
$(\pi:M\to U,g,\mathfrak{o},\mathfrak{t})$ such that 
vertical normally hyperbolic operators admit smoothly parametrized retarded and advanced Green operators
and a well-posed smoothly parametrized initial value problem.
\end{problem}

Successfully solving this problem will lead to a sensible definition
of a stack $\Loc_m^\infty$ of $m$-dimensional globally hyperbolic spacetimes.
One can then attempt to construct examples of smooth $m$-dimensional
AQFTs in terms of stack morphisms $\AAA : \Loc_m^\infty\to \astAlg^\infty$
by using the same strategy as in \eqref{eqn:AAAviaCCR}. 
We note that most of our constructions in Section \ref{sec:examples}
only rely on the existence (and uniqueness) of Green operators, hence they would generalize
directly to the higher-dimensional case, provided that Open Problem \ref{problem1}
is solved. There is however one exception: In the higher-dimensional case,
the space of initial data is infinite-dimensional, hence we can not argue
as in \eqref{eqn:Gdatadiagram} to conclude that the assignment
of linear observables $\LLL : \Loc_m^\infty\to \Pois^\infty$ is a stack morphism.
More specifically, this could lead to the problem that the coherence maps $\LLL_h$ 
are only natural transformations, but not natural isomorphisms,
which means that $\LLL$ is only a {\em lax} stack morphism. At the moment we do not know
whether it will be more convenient to enlarge the $2$-category
$\St$ of stacks to include also lax morphisms or to replace the stack $\Sh_{C^\infty_\bbR}$
of sheaves of $C^\infty_\bbR$-modules 
by a stack describing sheaves of topological (or bornological) modules 
in order to obtain a better control on these infinite-dimensional aspects.
Summing up, we record
\begin{problem}\label{problem2}
Find a suitable framework such that the assignment
$\LLL : \Loc_m^\infty\to \Pois^\infty$ of linear observables 
for a smooth $m$-dimensional free AQFT is a morphism between stacks.
Possible options could be enlarging the $2$-category $\St$ of stacks to allow for lax morphisms
or replacing the stack $\Sh_{C^\infty_\bbR}$ of sheaves of $C^\infty_\bbR$-modules 
by a stack describing sheaves of topological (or bornological) modules.
\end{problem}

As already emphasized at the beginning of Section \ref{sec:definition},
higher-dimensional AQFTs are sensitive to the phenomenon of Einstein causality,
which means that they are not simply functors but rather algebras over a suitable colored
operad \cite{BSWoperad,BSWinvolutive}. Encoding this aspect in our smooth setting
leads to the following
\begin{problem}\label{problem3}
Develop a theory of stacks of colored operads in order to define
the stack $\AQFT_m^\infty$  of smooth $m$-dimensional AQFTs
in terms of a suitable mapping stack between stacks of colored operads.
\end{problem}

To conclude, we would like to comment briefly 
on a potential generalization of our framework to gauge theories. 
The latter are most appropriately described by the BV-formalism,
which is captured by a concept of AQFTs taking values in cochain complexes, 
see e.g.\ \cite{FR1,FR2,BSWhomotopy}. This necessarily introduces to
an $\infty$-categorical context because the natural notion of
equivalence between cochain complexes is given by quasi-isomorphisms, 
as opposed to isomorphisms.
This in particular means that, instead of the smooth refinement 
$\Sh_{C^\infty_\bbK}$ of the ordinary category $\Vec_\bbK$ 
from Subsection \ref{subsec:Algstack}, one has to consider 
a smooth refinement of the {\em $\infty$-category} $\Ch_\bbK$
of cochain complexes.
A natural candidate for this purpose is the {\em $\infty$-stack}
$\Ch(\Sh_{C^\infty_\bbK})$ of cochain complexes of sheaves of modules.
\begin{problem}\label{problem4}
Replacing the stack $\Sh_{C^\infty_\bbK}$ of sheaves of modules
by the $\infty$-stack $\mathbf{Ch}(\Sh_{C^\infty_\bbK})$ of cochain complexes
of sheaves of modules, show that the relevant definitions 
and constructions from Section \ref{sec:definition} generalize to the context
of $\infty$-stacks.
\end{problem}

%%%%%%%%%%%%%%%%%%%%%%%%%%%%%%%%%%%%%%%%%%%%%%%%
%%%%%%%%%%%%%%%%%%%%%%%%%%%%%%%%%%%%%%%%%%%%%%%%

\section*{Acknowledgments}
We would like to thank Chris Fewster, Rune Haugseng, Fosco Loregian and Michael Shulman
for useful comments on this work.
We also would like to thank the anonymous referees for their comments
that helped us to improve the paper.
M.B.\ gratefully acknowledges the financial support of the 
National Group of Mathematical Physics GNFM-INdAM (Italy). 
M.P.\ is supported by a PhD scholarship (RGF\textbackslash EA\textbackslash 180270) 
of the Royal Society (UK). A.S.\ gratefully acknowledges the financial support of 
the Royal Society (UK) through a Royal Society University 
Research Fellowship (UF150099), a Research Grant (RG160517) 
and two Enhancement Awards (RGF\textbackslash EA\textbackslash 180270 
and RGF\textbackslash EA\textbackslash 201051).

%%%%%%%%%%%%%%%%%%%%%%%%%%%%%%%%%%%%%%%%%%%%%%%%
%%%%%%%%%%%%%%%%%%%%%%%%%%%%%%%%%%%%%%%%%%%%%%%%

%%%%%%%%%%%%%%%%%%%%%%%%

\end{document}